\documentclass[11pt]{article}
\usepackage[margin=1in]{geometry}
\usepackage{amsmath,amssymb,amsthm,amsfonts,bm,mathrsfs}
\usepackage{booktabs}
\usepackage{multirow}
\usepackage{natbib}
\usepackage{enumitem}
\usepackage{hyperref}
\usepackage{mathtools}
\usepackage{graphicx}
\usepackage{microtype}
\usepackage{caption}
\hypersetup{colorlinks=true,citecolor=blue,linkcolor=blue,urlcolor=blue}

\newtheorem{theorem}{Theorem}[section]
\newtheorem{proposition}{Proposition}[section]
\newtheorem{lemma}{Lemma}[section]

\newtheorem{assumption}{Assumption}[section]

\newcommand{\R}{\mathbb{R}}
\newcommand{\E}{\mathbb{E}}
\newcommand{\Pbb}{\mathbb{P}}
\newcommand{\1}{\mathbf{1}}
\newcommand{\norm}[1]{\left\lVert #1 \right\rVert}
\newcommand{\abs}[1]{\left\lvert #1 \right\rvert}
\newcommand{\argmin}{\mathop{\rm argmin}\limits}

\newcommand{\sgn}{\mathrm{sgn}}
\newcommand{\tr}{\mathrm{tr}}
\newcommand{\hbic}{\mathrm{HBIC}}
\newcommand{\maybeincludegraphics}[2][]{%
  \IfFileExists{#2}{\includegraphics[#1]{#2}}{\fbox{\texttt{\detokenize{#2}}}}%
}

\title{Rank-Based Sparse Regression in Principal Components Space under Measurement Error}
\author{Long Feng \\
{Department of Statistics and Data Science, Nankai University}\\
Xiaoyi Wang\\
Center for Statistics and Data Science, Beijing Normal University\\ 
Le Zhou \\
{Department of Mathematics, Hong Kong Baptist University} }
\date{}

\begin{document}
\maketitle

\begin{abstract}
We study high-dimensional regression in principal components space when the predictors are observed with additive measurement error and the response errors may be heavy-tailed. The starting point is the $\ell_1$-penalized principal-components estimator of \cite{SongZou2026}, which enjoys a blessing-of-dimensionality phenomenon under predictor contamination but senstive for heavy-tailed data or outliers. We replace the squared loss by a Wilcoxon-type rank loss and then apply a one-step adaptive reweighting scheme to reduce the shrinkage bias of the initial $\ell_1$ fit. The resulting procedure combines robustness to heavy-tailed response errors with the contamination geometry induced by the empirical principal-components basis. Our main theorem gives a prediction bound for the fixed-$\lambda$ second-stage fitted mean.  Simulations show that the rank-based procedure is competitive under Gaussian noise and substantially more stable under heavy-tailed errors, especially when predictor contamination is present.
\end{abstract}

\noindent\textbf{Keywords:} principal components regression; measurement error; robust high-dimensional regression; Wilcoxon rank regression.

\section{Introduction}

Sparsity has been the organizing principle of high-dimensional regression for nearly three decades. The modern literature begins with the lasso of \citet{Tibshirani1996} and includes basis-pursuit and Dantzig-selector formulations, together with a large class of folded-concave refinements such as SCAD and MCP; see, among many others, \citet{ChenDonohoSaunders2001}, \citet{CandesTao2007}, \citet{FanLi2001}, \citet{Zhang2010}, and \citet{BickelRitovTsybakov2009}. In most of this literature, sparsity is imposed in the original coordinate system of the predictors. That viewpoint is natural and often effective, but it is not the only meaningful structural assumption. Whenever the regression function is expressed after a linear transformation of the covariates, sparsity can also be imposed in the transformed basis. The practical value of such transformed-space sparsity depends on the scientific problem and on how well the transformed coordinates align with the signal.

Principal components provide the most classical transformed basis for regression. Building on that idea, \citet{LangZou2020} studied $\ell_1$-penalized regression in the empirical principal-components space and showed that it can materially improve on classical principal component regression. Related empirical evidence appears in \citet{ChoiEtAl2020}, while \citet{SilinFan2022} provided a broader theoretical analysis of sparse regression in principal-components space under nonsparse signal configurations. These papers emphasize an important distinction: principal-components sparsity is not the same as classical principal component regression, which keeps only a small number of leading components and may miss predictive directions carried by lower-variance components.

A second line of work studies regression with error-contaminated covariates. Measurement error has long been a central topic in statistics; see the monograph of \citet{CarrollEtAl2006}. In high dimensions, most available procedures operate in the original predictor space and correct the estimating equations or loss function by explicitly using the covariance structure of the contamination. Representative examples include the nonconvex approach of \citet{LohWainwright2012}, the matrix-uncertainty selector of \citet{RosenbaumTsybakov2013}, and the convex CoCoLasso procedure of \citet{DattaZou2017}. These methods are designed for coefficient estimation in the presence of noisy predictors and are powerful when the contamination model is sufficiently well specified. They are less directly informative when the predictive model is built in a transformed space, or when one wishes to understand the effect of contamination without assuming access to a reliable estimate of the measurement-error covariance.

This issue was recently revisited by \citet{SongZou2026}, who analyzed $\ell_1$-penalized regression in the empirical principal-components space built from contaminated predictors. Under weak sparsity of the regression signal in the latent principal-components basis, they showed that the prediction error can attain the clean-data minimax rate up to an additional perturbation term depending on the eigengap and the contamination level. Their theory also reveals a blessing-of-dimensionality phenomenon: when the number of predictors is large, the price paid for measurement error can diminish rather than increase. That finding is conceptually different from the usual sparse measurement-error regression theory in the original coordinate system.

The estimator of \citet{SongZou2026}, however, is a classical $\ell_1$-penalized least-squares procedure. Squared loss is statistically efficient under light-tailed Gaussian errors, but it can be unstable under heavy-tailed response contamination. Robust high-dimensional variable-selection methods have therefore been developed around alternative loss functions, including quantile and least-absolute-deviation losses \citep{BelloniChernozhukovWang2011, BradicFanWang2011, Wang2013, FanFanBarut2014}, Huber-type losses \citep{FanLiWang2017, Loh2017, SunZhouFan2020}, minimum-distance and other robustified criteria \citep{LozanoMeinshausenYang2016, AvellaMedinaRonchetti2018, PrasadEtAl2020}, and rank-based procedures. Particularly relevant for the present paper is the Wilcoxon-type rank method of \citet{WangEtAl2020}, which combines a tuning-free first-stage rank lasso with a second-stage adaptive bias-reduction step and enjoys both robustness and high efficiency.

The purpose of this paper is to combine these two ideas. We study rank-based sparse regression in empirical principal-components space when the predictors are contaminated by additive measurement error. Our first-stage estimator replaces the squared loss of \citet{SongZou2026} by the Wilcoxon rank loss while retaining the same empirical principal-components geometry. A second, adaptively weighted stage then follows the local-linear approximation principle of \citet{ZouLi2008} and \citet{WangEtAl2020}. The resulting procedure is intended to be robust to heavy-tailed response errors, while preserving the predictor-contamination structure that drives the blessing-of-dimensionality phenomenon in the principal-components model.

In this paper, we provide a two-stage rank-based procedure for sparse regression in principal-components space under measurement error. The theory shows that the second-stage fitted mean inherits the same eigengap-dependent contamination term as in the least-squares analysis of \citet{SongZou2026}, but improves the estimation term to the oracle order associated with the active principal components. Numerical results demonstrate that the rank-based procedure is especially advantageous under heavy-tailed errors and contaminated predictors.

The rest of the paper is organized as follows. Section~\ref{sec:model} introduces the model, the rank-based pilot estimator, and the second-stage adaptive refinement. Section~\ref{sec:theory} presents the assumptions and the main prediction theorem. Section~\ref{sec:simulation} reports simulation results and Section \ref{sec:data} gives a real data application. Section~\ref{sec:discussion} concludes with a brief discussion. Technical proofs are collected in the appendix.

\section{Methodology}
\label{sec:model}

\subsection{Principal-components-space model under contaminated predictors}

Write the singular-value decomposition of the latent design matrix as
\[
\mathbf X = \mathbf U\mathbf D\mathbf V^\top,
\]
where $m=\min(n,p)$, $\mathbf U\in\R^{n\times m}$ and $\mathbf V\in\R^{p\times m}$ have orthonormal columns, and $\mathbf D$ is diagonal with nonincreasing nonnegative entries. Following \citet{SongZou2026}, we parameterize the latent mean vector in the principal-components basis of $\mathbf X$:
\[
\bm y = \mathbf U\bm\gamma^* + \bm\varepsilon,
\qquad
\bm\theta^* = n^{-1/2}\bm\gamma^*.
\]
The vector $\bm\theta^*$ is the weakly sparse parameter of interest. Working on the $\bm\theta$-scale is convenient because the natural penalty level is then of order $\sqrt{\log m/n}$.

The latent predictor matrix is unobservable. Instead, we observe
\[
\mathbf Z = \mathbf X + \mathbf W,
\]
where $\mathbf W\in\R^{n\times p}$ is a measurement-error matrix. Let
\[
\mathbf Z = \widehat{\mathbf U}\widehat{\mathbf D}\widehat{\mathbf V}^\top
\]
be the singular-value decomposition of $\mathbf Z$, with columns ordered by decreasing singular values. We denote the $i$th row of $\widehat{\mathbf U}$ by $\widehat{\bm u}_i^\top$ and write
\[
\widetilde{\mathbf U}=\sqrt n\,\widehat{\mathbf U},
\qquad
\widetilde{\bm u}_i^\top = \sqrt n\,\widehat{\bm u}_i^\top.
\]
All estimators below are computed in this empirical principal-components basis.

Following \citet{SongZou2026}, a natural benchmark estimator in the empirical principal-components basis is the $\ell_1$-penalized least-squares estimator
\[
\widehat{\bm\theta}^{\mathrm{LS}}(\lambda)
=
\argmin_{\bm\theta\in\R^m}
\left\{
\frac{1}{2n}\bigl\|\bm y-\widetilde{\mathbf U}\bm\theta\bigr\|_2^2
+
\lambda\|\bm\theta\|_1
\right\},
\qquad
\widehat{\bm y}^{\mathrm{LS}}(\lambda)
=
\widetilde{\mathbf U}\widehat{\bm\theta}^{\mathrm{LS}}(\lambda).
\]
This estimator exploits weak sparsity in the empirical principal-components space and, in the least-squares framework, leads to the blessing-of-dimensionality phenomenon emphasized by \citet{SongZou2026}. In particular, the effect of predictor contamination enters through the perturbation of the empirical principal-components basis, and the resulting prediction error is controlled by an eigengap-dependent term.

However, the criterion above is still based on squared loss, and therefore its performance may deteriorate when the model errors are heavy-tailed or when a small fraction of the responses exhibits unusually large fluctuations. This limitation is especially relevant in high-dimensional problems with contaminated predictors, where one would like to retain the principal-components representation of the signal while reducing the sensitivity of the estimation step to extreme observations. These considerations motivate us to replace the least-squares loss by a rank-based loss and to combine it with a bias-reduction step.

\subsection{Stage 1: the rank-lasso pilot}

Define the pairwise rank loss
\begin{equation}
\mathcal Q_n(\bm\theta;\widehat{\mathbf U})
=
\frac{1}{n(n-1)}\sum_{i\neq j}
\Big|
(y_i-y_j)- (\widetilde{\bm u}_i-\widetilde{\bm u}_j)^\top\bm\theta
\Big|.
\label{eq:Qn}
\end{equation}
The first-stage estimator is
\begin{equation}
\widehat{\bm\theta}^{(0)}(\lambda_0)
=
\argmin_{\bm\theta\in\R^m}
\left\{
\mathcal Q_n(\bm\theta;\widehat{\mathbf U}) + \lambda_0\norm{\bm\theta}_1
\right\}.
\label{eq:stage1}
\end{equation}
This is a convex optimization problem and can be written as a linear program by introducing the usual slack variables for each pairwise residual. The associated fitted mean is
\[
\widehat{\bm y}^{(0)} = \widetilde{\mathbf U}\widehat{\bm\theta}^{(0)}(\lambda_0).
\]

For the pilot penalty level we adopt the simulated calibration of \citet{WangEtAl2020}. Conditionally on $\widehat{\mathbf U}$, generate a random permutation $\bm r=(r_1,\ldots,r_n)^\top$ of $(1,\ldots,n)$ and define
\[
\bm\xi = 2\bm r-(n+1)\1_n,
\qquad
\widehat{\mathbf S}_n^{\circ}= -\frac{2}{n(n-1)}\widetilde{\mathbf U}^\top \bm\xi.
\]
Let $G_{\|\widehat{\mathbf S}_n^{\circ}\|_\infty}$ denote the conditional distribution function of $\|\widehat{\mathbf S}_n^{\circ}\|_\infty$ given $\widehat{\mathbf U}$. For fixed constants $c>1$ and $\alpha_0\in(0,1)$, set
\begin{equation}
\widehat\lambda_0
=
 c\,G^{-1}_{\|\widehat{\mathbf S}_n^{\circ}\|_\infty}(1-\alpha_0).
\label{eq:lambda0-sim}
\end{equation}
The theoretical development below only requires that $\lambda_0\asymp\sqrt{\log m/n}$, but \eqref{eq:lambda0-sim} provides a convenient data-driven calibration in practice.

\subsection{Stage 2: adaptive reweighting and HBIC selection}

As in \citet{WangEtAl2020}, the rank-lasso pilot can be bias-corrected through a local-linear approximation to a folded-concave penalty. Let $p_\lambda(t)$ be a folded-concave penalty such as SCAD \citep{FanLi2001} or MCP \citep{Zhang2010}, and let $p_\lambda'(t)$ denote its derivative on $(0,\infty)$. The only properties used in the theory are
\begin{equation}
\text{$p_\lambda'(t)$ is decreasing,}\qquad
p_\lambda'(t)\ge a_1\lambda \ \text{for }0<t<a_2\lambda,
\qquad
p_\lambda'(t)=0\ \text{for }t>a_2\lambda,
\label{eq:penalty-condition}
\end{equation}
for some constants $a_1>0$ and $a_2>1$. The SCAD penalty function is given by
\[
p_{\lambda}(|\beta|)
=
\lambda |\beta| \, I(0 \le |\beta| < \lambda)
+
\frac{a\lambda |\beta| - (\beta^{2}+\lambda^{2})/2}{a-1}
\, I(\lambda \le |\beta| \le a\lambda)
+
\frac{(a+1)\lambda^{2}}{2}
\, I(|\beta| > a\lambda),
\]
for some \(a>2\). The MCP function has the form
\[
p_{\lambda}(|\beta|)
=
\lambda
\left(
|\beta| - \frac{\beta^{2}}{2a\lambda}
\right)
I(0 \le |\beta| < a\lambda)
+
\frac{a}{2}\lambda^{2}
\, I(|\beta| \ge a\lambda),
\]
for some \(a>1\). In practice, these two popular choices lead to similar performance.

Given the pilot estimator, define adaptive weights
\[
w_j(\lambda)=p_\lambda'\big(\abs{\widehat\theta_j^{(0)}}\big),
\qquad j=1,\ldots,m.
\]
The second-stage estimator is
\begin{equation}
\widehat{\bm\theta}^{(1)}(\lambda)
=
\argmin_{\bm\theta\in\R^m}
\left\{
\mathcal Q_n(\bm\theta;\widehat{\mathbf U}) + \sum_{j=1}^m w_j(\lambda)\abs{\theta_j}
\right\}.
\label{eq:stage2}
\end{equation}
For fixed weights this is again a convex weighted $\ell_1$ rank-regression problem.

To choose $\lambda$ in practice, let $\Lambda_1$ be a finite grid. For each $\lambda\in\Lambda_1$, compute $\widehat{\bm\theta}^{(1)}(\lambda)$ and its support
\[
A^{(1)}_{\lambda}=\{j:\widehat\theta_j^{(1)}(\lambda)\neq 0\}.
\]
Then refit the unpenalized rank loss on that support,
\[
\widetilde{\bm\theta}^{(1)}_{\lambda}
=
\argmin_{\bm\theta\in\R^m,\,\theta_{(A^{(1)}_{\lambda})^c}=0}
\mathcal Q_n(\bm\theta;\widehat{\mathbf U}),
\]
and define the high-dimensional BIC criterion
\begin{equation}
\hbic_1(\lambda)
=
\log\!\Big\{\mathcal Q_n\big(\widetilde{\bm\theta}^{(1)}_{\lambda};\widehat{\mathbf U}\big)\Big\}
+
|A^{(1)}_{\lambda}|\,\frac{\log\log n}{n}\log m.
\label{eq:hbic1}
\end{equation}
The practical second-stage penalty level is $\widehat\lambda = \argmin_{\lambda\in\Lambda_1}\hbic_1(\lambda)$, and the corresponding fitted mean is
\[
\widehat{\bm y}^{(1)}_{\mathrm{HBIC}}=\widetilde{\mathbf U}\widehat{\bm\theta}^{(1)}(\widehat\lambda).
\]

\section{Theory}
\label{sec:theory}

\subsection{Assumptions}

\begin{assumption}[Centering and coherence]
\label{ass:center}
The response and the principal-components bases are centered in the sense that $\1_n^\top\bm y=0$, $\1_n^\top\mathbf U=\bm 0$, and $\1_n^\top\widehat{\mathbf U}=\bm 0$. Moreover, there exists a constant $C_U>0$ such that
\[
\max_{1\le i\le n,\,1\le k\le m} \sqrt n |U_{ik}| \le C_U,
\qquad
\max_{1\le i\le n,\,1\le k\le m} \sqrt n |\widehat U_{ik}| \le C_U
\]
with probability tending to one.
\end{assumption}

\begin{assumption}[Heavy-tailed response errors]
\label{ass:error}
The variables $\varepsilon_1,\ldots,\varepsilon_n$ are independent and identically distributed, independent of $\mathbf W$, and have median zero. Their density $f_\varepsilon$ is continuous. Let
\[
g(t)=\int_{\R} f_\varepsilon(u)f_\varepsilon(u+t)\,du,
\]
which is the density of $\varepsilon_i-\varepsilon_j$ at $t$. There exist constants $g_0>0$, $L_g>0$, and $r_0>0$ such that
\[
g(0)\ge g_0,
\qquad
|g(t)-g(0)|\le L_g|t|,
\qquad |t|\le r_0.
\]
No finite-variance assumption is imposed on $\varepsilon_i$.
\end{assumption}

\begin{assumption}[Sub-Gaussian measurement errors]
\label{ass:W}
The latent design matrix $\mathbf X$ is fixed throughout. The rows of $\mathbf W$ are independent and identically distributed mean-zero $\tau_W^2$-sub-Gaussian random vectors. In addition, the squared row norms $\|\mathbf W_{i\cdot}\|_2^2$, $i=1,\ldots,n$, are sub-exponential with parameters of the same order as in \citet{SongZou2026}. The matrix $\mathbf W$ is independent of $\bm\varepsilon$. Let
\[
c_W = p^{-1}\E\|\mathbf W_{1\cdot}\|_2^2.
\]
\end{assumption}

Let
\[
\tau_\star=\sqrt{\frac{\log n}{n}},
\qquad
A=\left\{k\in[m]:|\theta_k^*|\ge \tau_\star\right\}.
\]
We focus on the nontrivial case $A\neq\varnothing$.
Let $\lambda_0\asymp\sqrt{\log m/n}$ and define
\begin{equation}
\Delta_{n,\kappa}
=
\frac{C_\Delta}{\kappa}
\left(
\frac{n}{p}+\frac{n^2}{p^2}+\frac{n\norm{\mathbf X}_2^2}{p^2}
\right)^{1/2},
\label{eq:Delta}
\end{equation}
where $C_\Delta$ depends only on the sub-Gaussian and sub-exponential constants in Assumption~\ref{ass:W}. Further let
\begin{equation}
\eta_n
=
\sqrt{|A|}\,\lambda_0
+
\sqrt{\frac{\log m}{n}}
+
R_q^{1/2}\tau_\star^{1-q/2}
+
\Delta_{n,\kappa}\frac{\norm{\bm y^*}_2}{\sqrt n}.
\label{eq:eta-n}
\end{equation}
 Every active eigenvalue of $p^{-1}\mathbf X\mathbf X^\top$ is assumed to be simple and separated from the rest:
\begin{equation}
\kappa
=
\min_{j\in A}\min_{\ell\neq j}
\left|
\zeta_j\left(\frac{1}{p}\mathbf X\mathbf X^\top\right)-\zeta_\ell\left(\frac{1}{p}\mathbf X\mathbf X^\top\right)
\right|
>0,
\label{eq:kappa-def}
\end{equation}
where $\zeta_i(\cdot)$ denotes the $i$th largest eigenvalue. When defining the columns of $\widehat{\mathbf U}$, choose the sign of $\widehat{\bm u}_j$ so that $\widehat{\bm u}_j^\top\bm u_j\ge 0$ for every $j\in A$.

\begin{assumption}[Signal structure, active eigengap, and tuning]
\label{ass:sparse}
For some fixed $q\in[0,1]$ and radius $R_q>0$,
\[
\sum_{k=1}^m |\theta_k^*|^q \le R_q.
\]
In addition, $\norm{\bm y^*}_2^2/n$ is bounded away from zero and infinity, where $\bm y^*=\mathbf U\bm\gamma^*=\sqrt n\,\mathbf U\bm\theta^*$.
We also assume
\label{ass:perturb}
\begin{equation}
\frac{1}{\kappa^2}
\left(
\frac{n}{p}+\frac{n^2}{p^2}+\frac{n\norm{\mathbf X}_2^2}{p^2}
\right)=o(1)
\qquad\text{and}\qquad
\frac{\log n}{\kappa^2 p}=o(1).
\label{eq:gap-rate}
\end{equation}
For sufficiently large constants $C_\beta>0$ and $C_\lambda>0$, the stage-2 penalty level $\lambda$ satisfies
\begin{equation}
\min_{j\in A}|\theta_j^*|
\ge
 a_2\lambda + C_\beta\eta_n,
\qquad
\lambda \ge C_\lambda\eta_n.
\label{eq:oracle-event-cond}
\end{equation}
\end{assumption}

\subsection{Main result}

The first result controls the pilot estimator needed to construct the second-stage adaptive weights.

\begin{theorem}[First-stage prediction bound]
\label{thm:stage1-pred}
Suppose Assumptions~\ref{ass:center}--\ref{ass:sparse} hold. Let $\widehat{\bm\theta}^{(0)}$ be any solution of \eqref{eq:stage1} with tuning parameter $\lambda_0\asymp\sqrt{\log m/n}$. Then, with probability at least $1-C_1 m^{-C_2}$,
\begin{equation}
\frac{1}{n}\norm{\widetilde{\mathbf U}\widehat{\bm\theta}^{(0)}-\bm y^*}_2^2
\le
C|A|\lambda_0^2
+
CR_q\tau_\star^{2-q}
+
\frac{C\,\norm{\bm y^*}_2^2}{n\kappa^2}
\left(
\frac{n}{p}+\frac{n^2}{p^2}+\frac{n\norm{\mathbf X}_2^2}{p^2}
\right).
\label{eq:stage1-bound}
\end{equation}
Moreover,
\begin{equation}
\norm{\widehat{\bm\theta}^{(0)}-\bar{\bm\theta}}_2
\le
C\Big(\sqrt{|A|}\,\lambda_0 + n^{-1/2}\norm{\bm v_A}_2\Big),
\qquad
\bm v_A=(\mathbf I_n-\widehat{\mathbf P}_A)\bm y^*,
\label{eq:l2-stage1}
\end{equation}
where $\widehat{\mathbf P}_A=\widehat{\mathbf U}_A\widehat{\mathbf U}_A^\top$.
\end{theorem}

Let
\[
\bar{\bm\theta}
=
\begin{pmatrix}
 n^{-1/2}\widehat{\mathbf U}_A^\top\bm y^*\\[0.3em]
 \bm 0
\end{pmatrix}
\]
be the oracle target obtained by truncating to the active empirical principal-components subspace. The corresponding oracle estimator is
\begin{equation}
\widehat{\bm\theta}^{(o)}
=
\argmin_{\bm\theta\in\R^m,\,\theta_{A^c}=0}
\mathcal Q_n(\bm\theta;\widehat{\mathbf U}).
\label{eq:oracle-def}
\end{equation}
To state the oracle reduction precisely, define the pilot separation event
\begin{equation}
\mathcal E_{\mathrm{sep}}(\lambda)=
\left\{
\min_{j\in A}\abs{\widehat\theta_j^{(0)}}>a_2\lambda,
\quad
\max_{k\in A^c}\abs{\widehat\theta_k^{(0)}}<a_2\lambda
\right\},
\label{eq:Esep}
\end{equation}
and the oracle off-support score event
\begin{equation}
\mathcal E_{\mathrm{score}}(\lambda)
=
\left\{
\begin{array}{l}
\text{there exists }\bm\Delta^{(o)}\in\partial\mathcal Q_n(\widehat{\bm\theta}^{(o)};\widehat{\mathbf U})\text{ with }\bm\Delta_A^{(o)}=\bm 0,
\norm{\bm\Delta_{A^c}^{(o)}}_\infty<a_1\lambda
\end{array}
\right\}.
\label{eq:Escore}
\end{equation}

The next theorem is the deterministic oracle-reduction step. It shows that once the pilot separation event and the oracle score event occur at the same fixed penalty level $\lambda$, the weighted second-stage objective has exactly the same minimizers as the truncated oracle objective.

\begin{theorem}[Second-stage oracle equivalence]
\label{thm:oracle-reduction}
Assume Assumptions~\ref{ass:center}--\ref{ass:sparse} hold, and let the penalty derivative satisfy \eqref{eq:penalty-condition}. On the event
$$
\mathcal E_{\mathrm{sep}}(\lambda)\cap \mathcal E_{\mathrm{score}}(\lambda),
$$
the minimizer set of the weighted second-stage problem \eqref{eq:stage2} coincides with the minimizer set of the oracle problem \eqref{eq:oracle-def}. In particular, every second-stage minimizer is an oracle minimizer, and every oracle minimizer is also a second-stage minimizer.
\end{theorem}

Lemma~\ref{lem:oracle-events} in the appendix shows that the signal-strength and tuning clause \eqref{eq:oracle-event-cond} in Assumption~\ref{ass:sparse} implies
\[
\Pbb\big(\mathcal E_{\mathrm{sep}}(\lambda)\cap\mathcal E_{\mathrm{score}}(\lambda)\big)\to 1.
\]

\begin{theorem}[Prediction error of the fixed-$\lambda$ second-stage fitted mean]
\label{thm:main-stage2}
Suppose Assumptions~\ref{ass:center}--\ref{ass:sparse} hold and let the penalty derivative satisfy \eqref{eq:penalty-condition}. Then the fixed-$\lambda$ second-stage fitted mean
\[
\widehat{\bm y}^{(1)}(\lambda)=\widetilde{\mathbf U}\widehat{\bm\theta}^{(1)}(\lambda)
\]
satisfies
\begin{equation}
\frac{1}{n}\norm{\widehat{\bm y}^{(1)}(\lambda)-\bm y^*}_2^2
=
O_P\left\{
\frac{|A|}{n}
+
R_q\tau_\star^{2-q}
+
\frac{\norm{\bm y^*}_2^2}{n\kappa^2}
\left(
\frac{n}{p}+\frac{n^2}{p^2}+\frac{n\norm{\mathbf X}_2^2}{p^2}
\right)
\right\}.
\label{eq:main-stage2-bound}
\end{equation}
In particular, relative to the first-stage rank-lasso rate, the estimation term improves to the oracle order $|A|/n$, while the predictor-contamination contribution remains unchanged.
\end{theorem}

The bound in \eqref{eq:main-stage2-bound} decomposes naturally into an estimation term and a contamination term. The estimation part $|A|/n+R_q\tau_\star^{2-q}$ is the oracle rate associated with the truncated active principal-components model. The last term is the perturbation price paid for using the empirical principal-components basis extracted from contaminated predictors. Theorem~\ref{thm:main-stage2} therefore shows that the second-stage rank procedure preserves the blessing-of-dimensionality structure identified by \citet{SongZou2026}, while replacing the first-stage lasso estimation term by its oracle analogue.

\section{Numerical Experiments}
\label{sec:simulation}

We compare two convex procedures. The first is the $\ell_1$-penalized least-squares estimator in empirical principal-components space from \citet{SongZou2026}, which we label \emph{L1PCR}. The second is the second-stage Wilcoxon rank estimator \eqref{eq:stage2}, which we label \emph{RPCR} to match the legends in the figures. Here we only choose the MCP penalty \citep{Zhang2010} in our simulation studies.

We consider the same settings as \cite{SongZou2026}.

\paragraph{Model 1.}
Let $\mathbf M$ be an $n\times p$ matrix with independent $N(0,1)$ entries, and let $\mathbf U$ and $\mathbf V$ be the left and right singular-vector matrices of $\mathbf M$. With $m=\min(n,p)$ and $A=\{1,\ldots,7\}$, define
\[
\mathbf X = a\mathbf U_A\mathbf V_A^\top + b\mathbf U_{A^c}\mathbf V_{A^c}^\top,
\qquad
a=\Big(0.9\,\frac{np}{7}\Big)^{1/2},
\qquad
b=\Big(0.1\,\frac{np}{m-7}\Big)^{1/2}.
\]
Set
\[
\bm\theta^*=(0.483,0,0.029,0.019,0,0.126,0.009,0,\ldots,0)^\top,
\qquad
\bm\gamma^*=\sqrt n\,\bm\theta^*.
\]
Thus the signal is concentrated on the first seven principal components, which is the same sparse principal-components design used by \citet{SongZou2026}.

\paragraph{Model 2.}
With the same preliminary random matrix $\mathbf M$, let $A=\{m-5,\ldots,m\}$ and define
\[
\mathbf X = a\mathbf U_A\mathbf V_A^\top + b\mathbf U_{A^c}\mathbf V_{A^c}^\top,
\qquad
a=(\kappa p)^{1/2},
\qquad
b=(2\kappa p)^{1/2}.
\]
All coordinates of $\bm\theta^*$ are set to $0.003$ except the last six, which are $0.009$, $0.125$, $0.003$, $0.019$, $0.029$, and $0.482$. Again $\bm\gamma^*=\sqrt n\,\bm\theta^*$. This model is more challenging because the important signals lie near the bottom of the principal-components spectrum and the coefficient vector is only weakly sparse.

As in \citet{SongZou2026}, the observable design is $\mathbf Z=\mathbf X+\mathbf W$. We consider three contamination regimes: $\mathbf W=\mathbf 0$ (no measurement error), independent measurement error with $\bm W \sim N(0,\mathbf I_p)$, and correlated measurement error with $\bm W\sim N(0,\mathbf{\Sigma})$ and $\Sigma_{ij}=0.5^{|i-j|}$.

For every configuration we report the average prediction error $n^{-1}\|\widehat{\bm y}-\bm y^*\|_2^2$ over 1{,}000 Monte Carlo replications. The response errors are generated as follows: (i) $\varepsilon_i\sim N(0,1)$; (ii) $\varepsilon_i\sim t_{3}/\sqrt 3$ so that the $t_3$ law is standardized to unit variance; and (iii) $\varepsilon_i=\widetilde\varepsilon_i/\sqrt{10.9}$ with $\widetilde\varepsilon_i\sim 0.9N(0,1)+0.1N(0,100)$, which yields a standardized mixture-normal distribution. The latter two settings represent increasingly heavy-tailed response contamination.

Figures~\ref{fig:model1p} and \ref{fig:model2p} summarize the effect of increasing the number of predictors. The most stable qualitative feature is the same one emphasized by \citet{SongZou2026}: under predictor contamination, the average prediction error decreases as $p$ grows. This is visible for both L1PCR and RPCR in both models and across all three response-error distributions. In other words, the blessing-of-dimensionality phenomenon persists after the squared loss is replaced by the rank loss.

The distinction between the two methods is driven primarily by the tail behavior of the response errors. When there is no measurement error, RPCR is uniformly competitive and typically more accurate than L1PCR, with especially pronounced gains under the standardized $t_3$ and standardized mixture-normal errors. Once predictor contamination is introduced, the Gaussian setting becomes more nuanced: L1PCR often retains a modest efficiency advantage under normal errors, particularly in the contaminated cases of both models. This is unsurprising because least squares is well aligned with a light-tailed Gaussian response model.

The pattern changes markedly under heavy-tailed errors. Under standardized $t_3$ errors, RPCR typically improves on L1PCR in the contaminated designs, and the advantage becomes clearer as $p$ increases. Under standardized mixture-normal errors, the gains are substantial. Across both models and under both independent and correlated measurement error, the rank-based method yields visibly smaller prediction error, with the gap generally widening at larger dimensions. The improvement is especially strong in Model~2, where the signal is weaker and concentrated near the lower end of the spectrum.

\begin{figure}[htbp]
  \centering
  \maybeincludegraphics[width=0.98\textwidth]{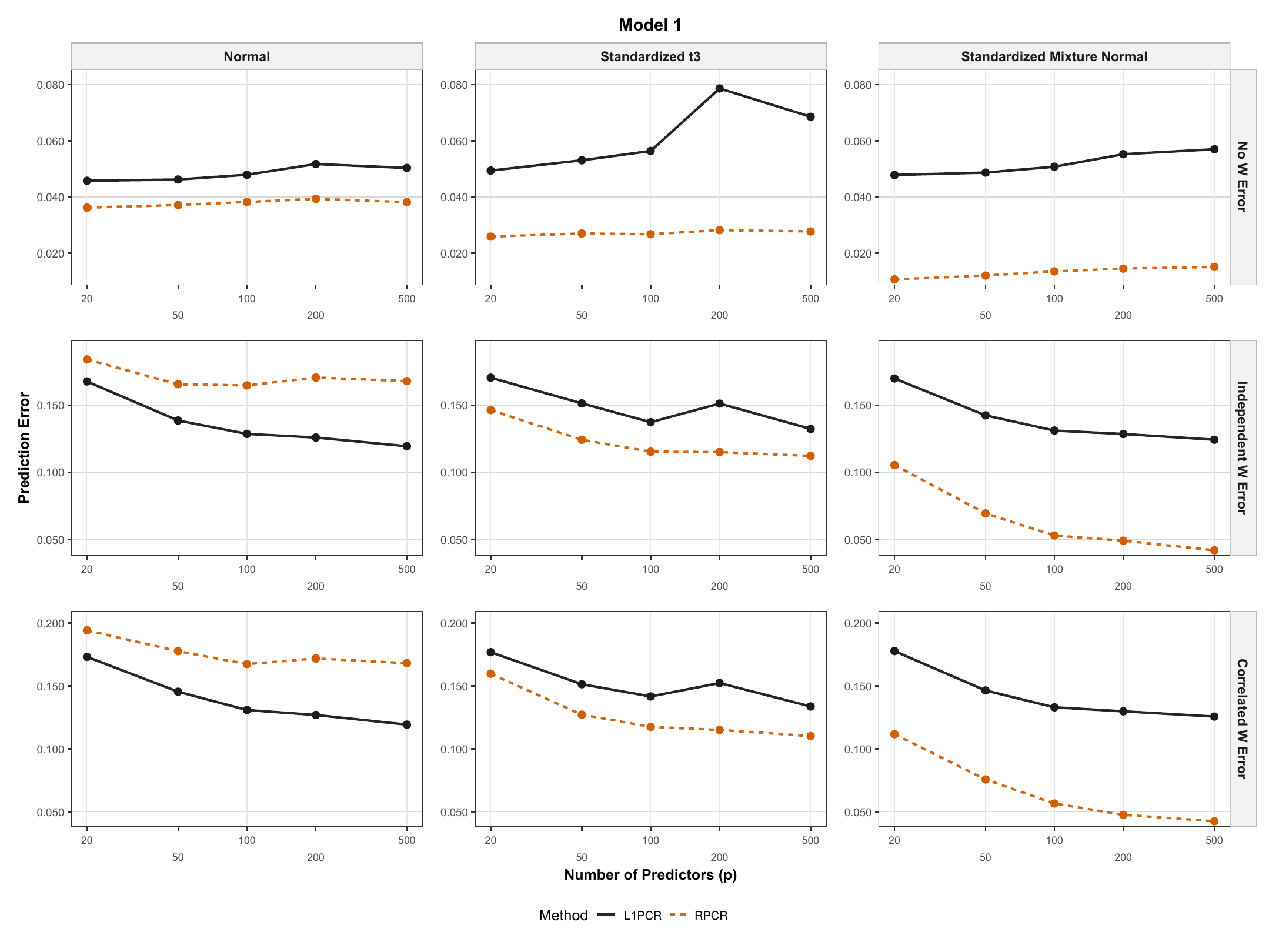}
  \caption{Model~1: average prediction error over 5{,}000 replications as a function of the number of predictors $p$. Rows correspond to no measurement error, independent measurement error, and correlated measurement error. Columns correspond to normal, standardized $t_3$, and standardized mixture-normal response errors.}
  \label{fig:model1p}
\end{figure}

\begin{figure}[htbp]
  \centering
  \maybeincludegraphics[width=0.98\textwidth]{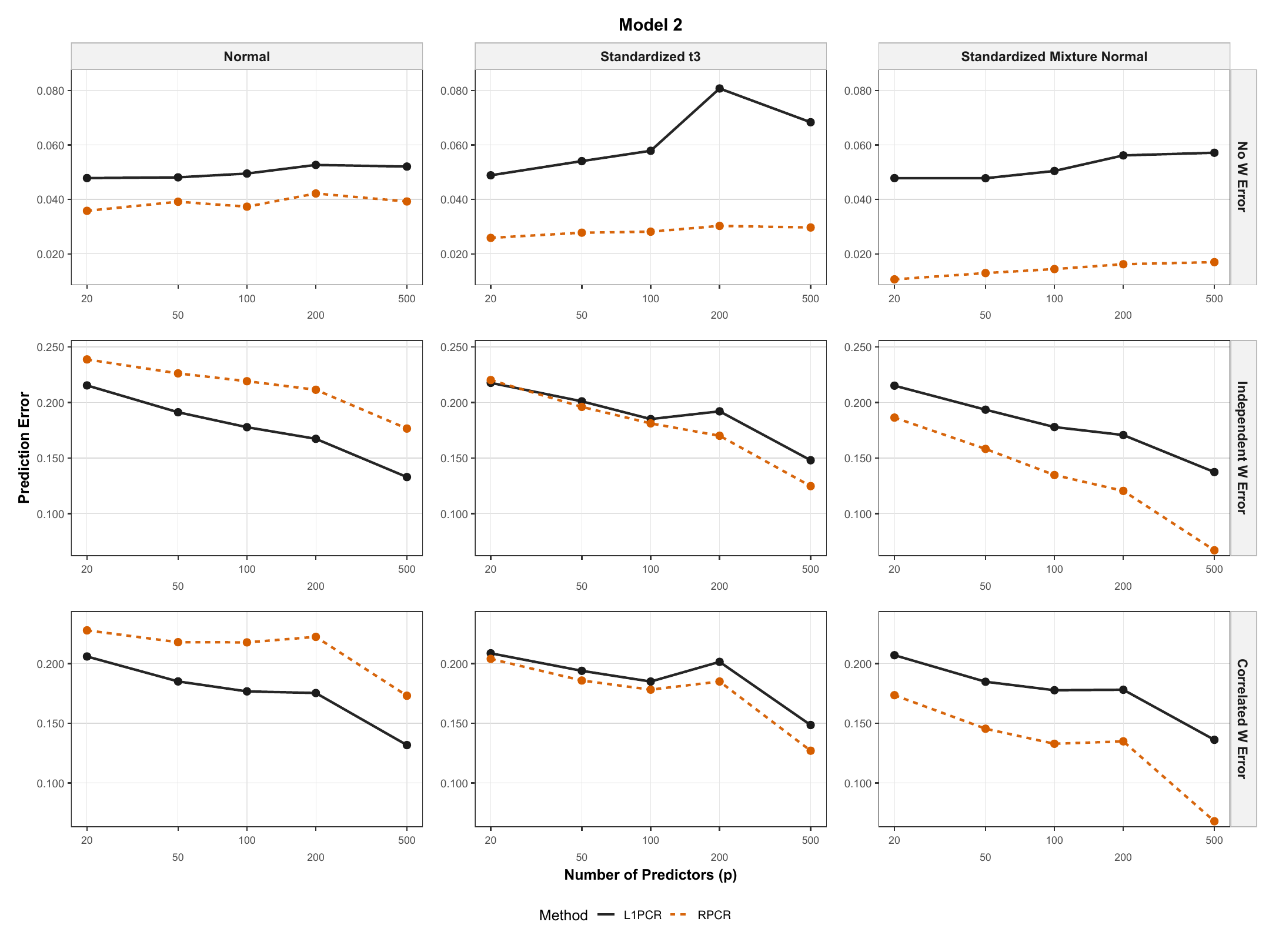}
  \caption{Model~2: average prediction error as a function of the number of predictors $p$. The design layout and error regimes are the same as in Figure~\ref{fig:model1p}.}
  \label{fig:model2p}
\end{figure}

The theory in Section~\ref{sec:theory} shows that the contamination term is inversely proportional to $\kappa^2$, so the eigengap should matter only when the predictors are contaminated. Figures~\ref{fig:model2k1} and \ref{fig:model2k2} confirm exactly that pattern. In the no-measurement-error row, the curves are essentially flat in $\kappa$ for both methods, which is consistent with the fact that the perturbation term disappears when $\mathbf W=\mathbf 0$. Under independent and correlated measurement error, however, increasing $\kappa$ lowers the prediction error for both L1PCR and RPCR.

The relative behavior of the two procedures again depends on tail thickness. Under normal response errors, the two methods are broadly comparable in the $\kappa$ experiments, and L1PCR is often slightly better in the contaminated settings. Under standardized $t_3$ errors, RPCR improves more quickly as $\kappa$ increases and eventually dominates in most contaminated cases. Under standardized mixture-normal errors, the difference is even sharper: the orange RPCR curve drops substantially faster than the black L1PCR curve, especially under correlated measurement error. The same qualitative message appears in both $\kappa$ plots, indicating that the interaction between a larger eigengap and a robust loss is stable across the two fixed-dimensional regimes.

\begin{figure}[htbp]
  \centering
  \maybeincludegraphics[width=0.98\textwidth]{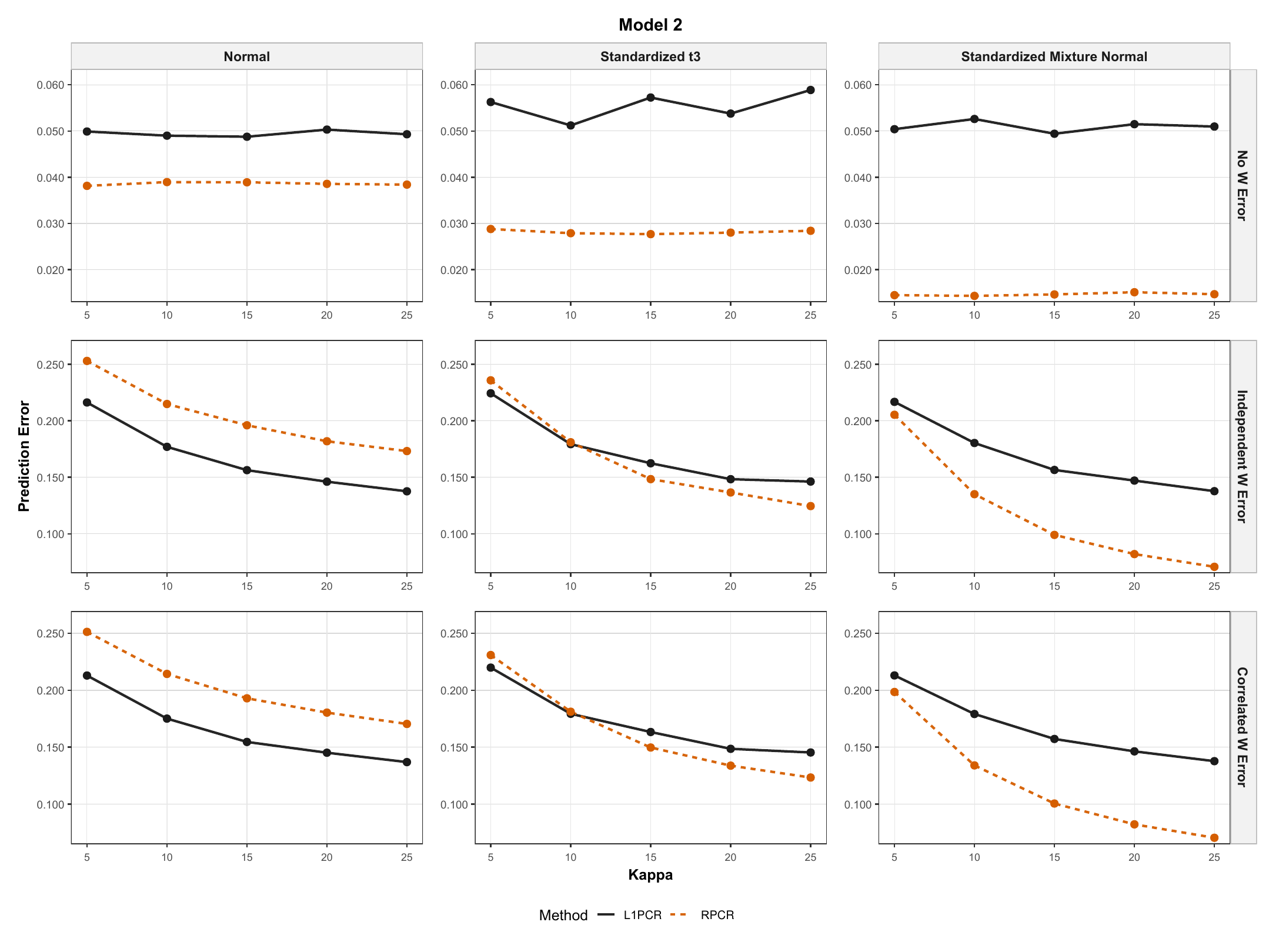}
  \caption{Model~2: average prediction error as a function of the eigengap parameter $\kappa$ in one representative fixed-$p$ regime. Rows and columns are organized as in Figure~\ref{fig:model1p}.}
  \label{fig:model2k1}
\end{figure}

\begin{figure}[htbp]
  \centering
  \maybeincludegraphics[width=0.98\textwidth]{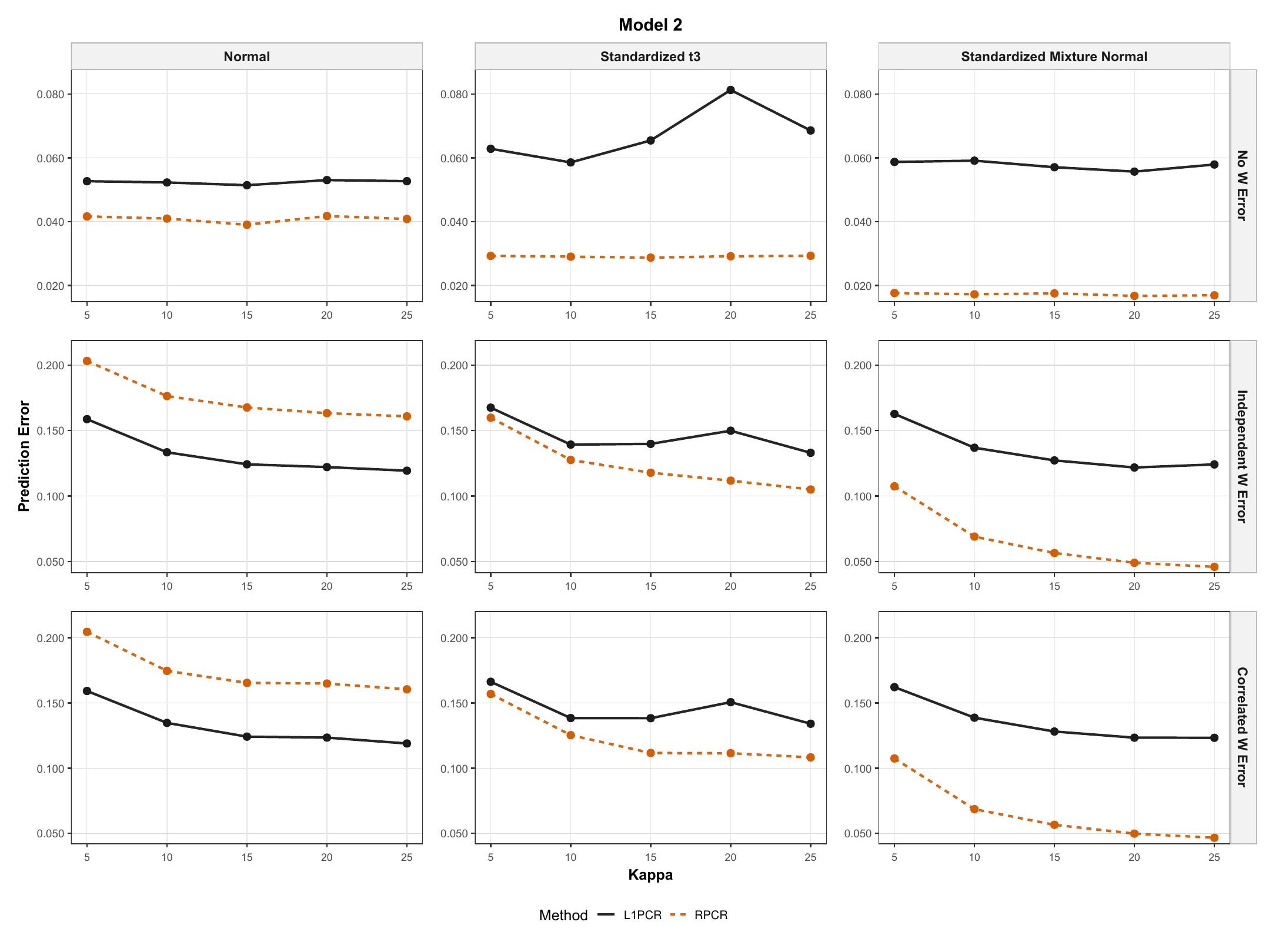}
  \caption{Model~2: average prediction error as a function of the eigengap parameter $\kappa$ in a second, larger fixed-$p$ regime. The qualitative behavior matches that of Figure~\ref{fig:model2k1}.}
  \label{fig:model2k2}
\end{figure}

\section{Real Data Application}\label{sec:data}

We next examine the finite-sample behavior of the proposed rank-based principal-components regression on a real genetic dataset. The data come from the eye eQTL study of \cite{Scheetz2006EyeDisease}, which was also used in the real-data illustration of \citet{WangEtAl2020}. The public \texttt{Scheetz2006} file contains normalized expression measurements for 120 twelve-week-old male rat offspring. Following \citet{WangEtAl2020}, we take the expression level of \texttt{TRIM32} as the response. In our sample the standard deviation of the response is $0.145$. The public file contains a much larger set of probes, so to match the 300-predictor setting used in the empirical illustration we retain the 300 probes having the largest absolute marginal correlation with \texttt{TRIM32}, and use these 300 probes as predictors.

Our data analysis follows the same evaluation strategy as the real-data experiment in \citet{SongZou2026}. Let $X$ denote the resulting $120 \times 300$ predictor matrix and let $y$ be the \texttt{TRIM32} response. For each contamination level $c \in \{0,0.1,0.2,0.3,0.4,0.5\}$, we generate a contaminated design matrix
\[
Z = X + W, \qquad W_{ij} \stackrel{\text{i.i.d.}}{\sim} N(0,\sigma^2),
\]
where
\[
\sigma = c \left\{\frac{1}{p}\sum_{j=1}^p \widehat{\mathrm{Var}}(X_j)\right\}^{1/2}.
\]
For each value of $c$, we evaluate prediction performance by leave-one-out cross-validation. In each split, the left singular vectors of the full contaminated design matrix are computed first, exactly as in \citet{SongZou2026}. We then compare three methods:
\begin{enumerate}
\item \textbf{L1PCR}: the $\ell_1$-penalized principal-components regression estimator of \citet{SongZou2026};
\item \textbf{LASSO}: the standard lasso fitted directly on the contaminated predictors;
\item \textbf{RPCR}: our proposed procedure, which first maps the data into the principal-components space and then applies stage-two rank lasso with an MCP refinement.
\end{enumerate}
For L1PCR and LASSO, the tuning parameter is chosen by ten-fold cross-validation within each training sample. For RPCR, we fit the rank-based estimator on the principal-components scores and use the stage-two MCP estimate for prediction. As in \citet{SongZou2026}, we report the average leave-one-out squared prediction error. We also report standard errors for pairwise differences in squared prediction errors across held-out observations.

\begin{table}[htbp]
\centering
\caption{Average leave-one-out squared prediction error for the Scheetz genetic dataset under additive measurement error. The columns labeled SE report the standard error of the corresponding pairwise difference in squared prediction error.}
\label{tab:scheetz_realdata}
\scriptsize
\begin{tabular}{lcccccc}
\hline
$c$ & L1PCR & LASSO & RPCR & SE(L1PCR-LASSO) & SE(L1PCR-RPCR) & SE(LASSO-RPCR) \\
\hline
0.0 & 0.00984 & \textbf{0.00684} & 0.01171 & 0.00206 & 0.00194 & 0.00391 \\
0.1 & 0.00817 & \textbf{0.00679} & 0.00703 & 0.00085 & 0.00140 & 0.00093 \\
0.2 & 0.00945 & 0.00830 & \textbf{0.00690} & 0.00090 & 0.00270 & 0.00209 \\
0.3 & 0.00809 & 0.00799 & \textbf{0.00676} & 0.00057 & 0.00146 & 0.00185 \\
0.4 & 0.00981 & 0.01086 & \textbf{0.00690} & 0.00096 & 0.00274 & 0.00357 \\
0.5 & 0.00977 & 0.00891 & \textbf{0.00742} & 0.00064 & 0.00260 & 0.00221 \\
\hline
\end{tabular}
\end{table}

Table~\ref{tab:scheetz_realdata} delivers the same qualitative message as our simulations. In the clean design ($c=0$) and under very mild contamination ($c=0.1$), LASSO achieves the smallest prediction error. However, once measurement error becomes moderate ($c \ge 0.2$), the proposed RPCR method uniformly attains the smallest average prediction error. The gain is particularly visible at the larger contamination levels. For example, when $c=0.4$, the prediction error of RPCR is $0.00690$, compared with $0.00981$ for L1PCR and $0.01086$ for LASSO. At $c=0.5$, RPCR still remains best, with prediction error $0.00742$. Thus, while the rank-based procedure is not uniformly best in the nearly clean design, it becomes the most stable method once the predictors are materially contaminated.

\begin{figure}[htbp]
\centering
\includegraphics[width=\textwidth]{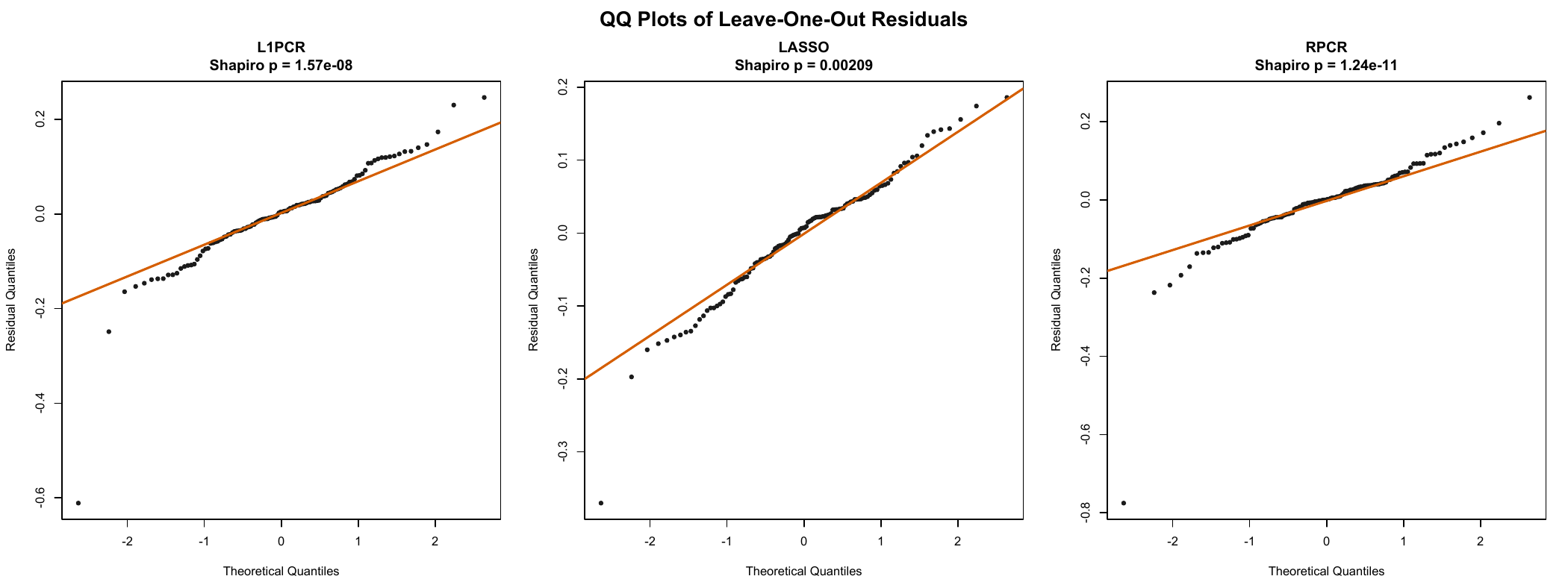}
\caption{QQ plots of leave-one-out residuals for L1PCR, LASSO, and RPCR at the clean design ($c=0$). The Shapiro--Wilk $p$-values are $1.57\times 10^{-8}$ for L1PCR, $0.0021$ for LASSO, and $1.24\times 10^{-11}$ for RPCR. All three tests reject normality, and the tail departures are especially pronounced for L1PCR and RPCR.}
\label{fig:scheetz_realdata_qq}
\end{figure}

Figure~\ref{fig:scheetz_realdata_qq} shows clear departures from Gaussianity in the leave-one-out residuals. All three Shapiro--Wilk $p$-values are below $0.01$, so the Gaussian error assumption is not well supported by the data. This diagnostic is important for interpretation: it indicates that a rank-based method is empirically well motivated in this application. Taken together, Table~\ref{tab:scheetz_realdata} and Figure~\ref{fig:scheetz_realdata_qq} suggest that the stage-two rank-based principal-components procedure is competitive in the nearly clean regime and becomes more reliable than both L1PCR and direct LASSO once measurement error is introduced into the predictors. This pattern is consistent with the main theme of the paper: robust estimation in principal-components space can deliver practical gains when the design is high-dimensional and the data depart from the ideal Gaussian benchmark.

\section{Discussion}
\label{sec:discussion}

This paper studies prediction in high-dimensional errors-in-variables regression when the signal is sparse in the empirical principal-components space. We propose a rank-based two-stage estimator that combines an $\ell_1$-penalized Wilcoxon pilot with a refitted second stage to reduce shrinkage bias while retaining the same eigengap-controlled contamination term as in the least-squares framework. Both the theory and the simulations show that the proposed method is especially advantageous under heavy-tailed errors, where it is substantially more stable and accurate than least-squares-based procedures.

Several extensions merit further investigation. One natural direction is to extend the present analysis beyond linear models to high-dimensional generalized linear models \citep{vandeGeer2008,JiangZhouLiuMa2023}, where the interaction among nonlinear links, rank-based loss functions, and predictor contamination would require new theory. It would also be of interest to study contaminated high-dimensional regression under other transformed predictor bases, such as principal-component-guided sparse regression \citep{TayFriedmanTibshirani2021}, so that the structural assumption can better adapt to the geometry of the signal. These questions are left for future work.

\appendix

\section{Auxiliary lemmas}

\begin{proposition}[Jaeckel representation]
\label{prop:jaeckel}
Let $\bm r=(r_1,\ldots,r_n)^\top$ be any residual vector. Denote by $R_i(\bm r)$ the rank of $r_i$ among $r_1,\ldots,r_n$. Then
$$
\sum_{i=1}^n \Big(R_i(\bm r)-\frac{n+1}{2}\Big) r_i
=
\frac12\sum_{1\le i<j\le n} |r_i-r_j|.
$$
Consequently, minimizing $\mathcal Q_n(\bm\theta;\widehat{\mathbf U})$ is equivalent to minimizing Jaeckel's dispersion with Wilcoxon scores computed from the empirical principal-components design $\widetilde{\mathbf U}$.
\end{proposition}

\begin{proof}
Let $r_{(1)}\le\cdots\le r_{(n)}$ be the order statistics. Then
$$
\sum_{i=1}^n \Big(R_i(\bm r)-\frac{n+1}{2}\Big)r_i
=
\sum_{k=1}^n \Big(k-\frac{n+1}{2}\Big)r_{(k)}.
$$
A direct summation-by-parts argument gives
$$
\sum_{k=1}^n \Big(k-\frac{n+1}{2}\Big)r_{(k)}
=
\frac12\sum_{1\le i<j\le n}(r_{(j)}-r_{(i)})
=
\frac12\sum_{1\le i<j\le n}|r_i-r_j|.
$$
Dividing by $n(n-1)$ only rescales the objective and does not change the minimizer.
\end{proof}

\begin{lemma}[Tail bounds for weak sparsity]
\label{lem:tail}
Let $\bm a\in\R^m$ satisfy $\sum_{k=1}^m |a_k|^q\le R_q$ with $q\in[0,1]$. For any $\tau>0$, define
$$
B(\tau)=\{k:|a_k|\ge \tau\}.
$$
Then
$$
|B(\tau)| \le R_q\tau^{-q},
\qquad
\sum_{k\notin B(\tau)} a_k^2 \le R_q\tau^{2-q},
\qquad
\sum_{k\notin B(\tau)} |a_k| \le R_q\tau^{1-q}.
$$
\end{lemma}

\begin{proof}
The cardinality bound follows from $|B(\tau)|\tau^q\le \sum_k|a_k|^q$. The other two bounds follow from
$$
\sum_{k\notin B(\tau)}|a_k|^r
=
\sum_{k\notin B(\tau)}|a_k|^q|a_k|^{r-q}
\le
\tau^{r-q}\sum_{k\notin B(\tau)}|a_k|^q,
$$
with $r=2$ and $r=1$, respectively.
\end{proof}

\begin{lemma}[Local quadratic curvature]
\label{lem:curvature}
Let $\xi=\varepsilon_1-\varepsilon_2$, where $\varepsilon_1$ and $\varepsilon_2$ are independent copies of $\varepsilon$. Under Assumption~\ref{ass:error}, there exist constants $0<c_\ell<c_u<\infty$ and $t_0\in(0,r_0]$ such that for all $|t|\le t_0$,
$$
c_\ell t^2
\le
\E\big\{ |\xi-t|-|\xi| \big\}
\le
c_u t^2.
$$
\end{lemma}

\begin{proof}
Define $\phi(t)=\E\{|\xi-t|-|\xi|\}$. Since the density of $\xi$ is $g$, we have $\phi'(0)=0$ and $\phi''(t)=2g(t)$ for $|t|<r_0$. By symmetry of $\xi$,
$$
\phi(t)=2\int_0^{|t|} (|t|-s)g(s)\,ds.
$$
Using the Lipschitz bound on $g$, choose $t_0$ so that $g_0-L_g t_0\ge g_0/2$ and $g_0+L_g t_0\le 3g_0/2$. Then for $|t|\le t_0$,
$$
\frac{g_0}{2}t^2 \le \phi(t) \le \frac{3g_0}{2}t^2.
$$
This gives the result.
\end{proof}

\begin{lemma}[Score concentration]
\label{lem:score}
Let
$$
\mathbf S_n
=
\frac{1}{n(n-1)}\sum_{i\ne j}
(\widetilde{\bm u}_i-\widetilde{\bm u}_j)
\sgn(\varepsilon_i-\varepsilon_j).
$$
Under Assumptions~\ref{ass:center} and \ref{ass:error}, there exist positive constants $C_1$ and $C_2$ such that
$$
\Pbb\bigg( \norm{\mathbf S_n}_\infty > C_1\sqrt{\frac{\log m}{n}} \bigg)
\le 2m^{-C_2}.
$$
\end{lemma}

\begin{proof}
Condition on $\widehat{\mathbf U}$. For each coordinate $k$, the kernel
$$
h_k\big((\varepsilon_i,\widetilde{\bm u}_i),(\varepsilon_j,\widetilde{\bm u}_j)\big)
=
(\widetilde U_{ik}-\widetilde U_{jk})\sgn(\varepsilon_i-\varepsilon_j)
$$
is bounded by $2C_U$ under Assumption~\ref{ass:center}. Hoeffding's concentration inequality for bounded U-statistics therefore yields
$$
\Pbb\big( |S_{n,k}|>t \mid \widehat{\mathbf U} \big)
\le 2\exp\Big( -\frac{c n t^2}{C_U^2} \Big)
$$
for some universal constant $c>0$. A union bound over $k=1,\ldots,m$ gives the claim.
\end{proof}

\begin{lemma}[Permutation representation of the score]
\label{lem:perm-score}
Let $R_i$ be the rank of $\varepsilon_i$ among $\varepsilon_1,\ldots,\varepsilon_n$, and define
$$
\eta_i = 2R_i-(n+1),
\qquad
\bm\eta=(\eta_1,\ldots,\eta_n)^\top.
$$
Then
$$
\mathbf S_n
=
\frac{2}{n(n-1)}\widetilde{\mathbf U}^\top\bm\eta,
\qquad
\mathbf S_{n,A}
=
\frac{2}{n(n-1)}\widetilde{\mathbf U}_A^\top\bm\eta.
$$
Conditional on $\widehat{\mathbf U}$, the vector $\bm\eta$ is a uniform random permutation of
$$
(-(n-1),-(n-3),\ldots,n-3,n-1)^\top.
$$
Moreover,
\begin{equation}
\E\!\left(\bm\eta\bm\eta^\top \mid \widehat{\mathbf U}\right)
=
\frac{n(n+1)}{3}
\left(
\mathbf I_n-\frac{1}{n}\1_n\1_n^\top
\right).
\label{eq:eta-cov}
\end{equation}
\end{lemma}

\begin{proof}
For each $i$,
$$
\sum_{j\ne i}\sgn(\varepsilon_i-\varepsilon_j)
=
(R_i-1)-(n-R_i)
=
2R_i-(n+1)
=
\eta_i.
$$
Therefore
\begin{align*}
\mathbf S_n
&=
\frac{1}{n(n-1)}\sum_{i\ne j}
(\widetilde{\bm u}_i-\widetilde{\bm u}_j)\sgn(\varepsilon_i-\varepsilon_j)\\
&=
\frac{2}{n(n-1)}\sum_{i=1}^n \widetilde{\bm u}_i
\sum_{j\ne i}\sgn(\varepsilon_i-\varepsilon_j)
=
\frac{2}{n(n-1)}\widetilde{\mathbf U}^\top\bm\eta.
\end{align*}
Restricting the coordinates to $A$ gives the same formula for $\mathbf S_{n,A}$.

Because the errors are i.i.d. with continuous density, every ordering of $\varepsilon_1,\ldots,\varepsilon_n$ is equally likely. Hence $\bm\eta$ is a uniform random permutation of the deterministic centered score vector displayed above. Its coordinates have mean zero, marginal variance $(n^2-1)/3$, and pairwise covariance $-(n+1)/3$. This is exactly the matrix formula in \eqref{eq:eta-cov}.
\end{proof}

\begin{lemma}[Restricted score norm]
\label{lem:scoreA}
Let
$$
\mathbf S_{n,A}
=
\frac{1}{n(n-1)}\sum_{i\ne j}
(\widetilde{\bm u}_{i,A}-\widetilde{\bm u}_{j,A})
\sgn(\varepsilon_i-\varepsilon_j).
$$
Under Assumptions~\ref{ass:center} and \ref{ass:error}, there exist positive constants $C_1$ and $C_2$ such that
$$
\Pbb\bigg( \norm{\mathbf S_{n,A}}_2 > C_1\sqrt{\frac{|A|}{n}} \bigg)
\le 2e^{-C_2|A|}.
$$
\end{lemma}

\begin{proof}
Condition on $\widehat{\mathbf U}$ and write
$$
\mathbf B_A=\frac{2}{n(n-1)}\widetilde{\mathbf U}_A^\top,
\qquad
\mathbf S_{n,A}=\mathbf B_A\bm\eta
$$
by Lemma~\ref{lem:perm-score}. Since $\1_n^\top\widetilde{\mathbf U}_A=\bm 0^\top$, the centering projector in \eqref{eq:eta-cov} acts as the identity on the column space of $\widetilde{\mathbf U}_A$. Therefore
\begin{align*}
\E\!\left(\norm{\mathbf S_{n,A}}_2^2 \mid \widehat{\mathbf U}\right)
&=
\tr\!\left\{
\mathbf B_A
\E(\bm\eta\bm\eta^\top\mid \widehat{\mathbf U})
\mathbf B_A^\top
\right\}\\
&=
\frac{n(n+1)}{3}\tr(\mathbf B_A\mathbf B_A^\top)\\
&=
\frac{n(n+1)}{3}\cdot
\frac{4}{n^2(n-1)^2}
\tr(\widetilde{\mathbf U}_A^\top\widetilde{\mathbf U}_A)\\
&=
\frac{4(n+1)}{3(n-1)^2}|A|
\le
C\frac{|A|}{n}.
\end{align*}
Hence
$$
\E\!\left(\norm{\mathbf S_{n,A}}_2 \mid \widehat{\mathbf U}\right)
\le
C\sqrt{\frac{|A|}{n}}.
$$

Next, the map $f(\bm z)=\norm{\mathbf B_A\bm z}_2$ is convex and Lipschitz with constant
$$
\|\mathbf B_A\|_{\mathrm{op}}
=
\frac{2}{n(n-1)}\norm{\widetilde{\mathbf U}_A}_{\mathrm{op}}
=
\frac{2\sqrt n}{n(n-1)}
\le
\frac{C}{n^{3/2}}.
$$
Since $\bm\eta$ is a random permutation of a vector with entries in $[-(n-1),n-1]$, the standard concentration inequality for convex Lipschitz functions on the symmetric group yields
$$
\Pbb\Big(
\norm{\mathbf S_{n,A}}_2
>
\E(\norm{\mathbf S_{n,A}}_2\mid \widehat{\mathbf U}) + t
\,\Big|\,
\widehat{\mathbf U}
\Big)
\le
2\exp(-c n t^2)
$$
for all $t>0$, where $c>0$ is universal. Taking $t=C\sqrt{|A|/n}$ and then integrating out $\widehat{\mathbf U}$ proves the claim.
\end{proof}

\begin{lemma}[Pairwise design identity]
\label{lem:pairwise-id}
Let $\mathbf X\in\R^{n\times d}$ satisfy $\1_n^\top\mathbf X=\bm 0^\top$, and let $\bm x_i^\top$ denote its $i$th row. Then
$$
\sum_{i\ne j}(\bm x_i-\bm x_j)(\bm x_i-\bm x_j)^\top
=
2n\,\mathbf X^\top\mathbf X.
$$
Consequently, for any $\bm\delta\in\R^d$,
$$
\frac{1}{n(n-1)}\sum_{i\ne j}\big\{(\bm x_i-\bm x_j)^\top\bm\delta\big\}^2
=
\frac{2}{n-1}\,\bm\delta^\top\mathbf X^\top\mathbf X\,\bm\delta.
$$
In particular, when $\mathbf X=\widetilde{\mathbf U}$ or $\widetilde{\mathbf U}_A$,
$$
\frac{1}{n(n-1)}\sum_{i\ne j}\big\{(\widetilde{\bm u}_i-\widetilde{\bm u}_j)^\top\bm\delta\big\}^2
=
\frac{2n}{n-1}\norm{\bm\delta}_2^2 .
$$
\end{lemma}

\begin{proof}
Expand the square and use $\sum_i \bm x_i=\bm 0$:
\begin{align*}
\sum_{i\ne j}(\bm x_i-\bm x_j)(\bm x_i-\bm x_j)^\top
&=
\sum_{i\ne j}\left(\bm x_i\bm x_i^\top+\bm x_j\bm x_j^\top-\bm x_i\bm x_j^\top-\bm x_j\bm x_i^\top\right)\\
&=2(n-1)\sum_{i=1}^n \bm x_i\bm x_i^\top - 2\sum_{i\ne j}\bm x_i\bm x_j^\top\\
&=2(n-1)\mathbf X^\top\mathbf X + 2\mathbf X^\top\mathbf X
=2n\,\mathbf X^\top\mathbf X.
\end{align*}
The stated consequence follows by contracting with $\bm\delta$ on both sides.
\end{proof}

\begin{lemma}[Maximum pairwise increment]
\label{lem:maxpair}
Under Assumption~\ref{ass:center}, for any $\bm\delta\in\R^m$,
$$
\max_{i\ne j} \big| (\widetilde{\bm u}_i-\widetilde{\bm u}_j)^\top\bm\delta \big|
\le
2C_U\norm{\bm\delta}_1.
$$
\end{lemma}

\begin{proof}
For any $i\ne j$,
$$
\big| (\widetilde{\bm u}_i-\widetilde{\bm u}_j)^\top\bm\delta \big|
\le
\sum_{k=1}^m |\widetilde U_{ik}-\widetilde U_{jk}|\,|\delta_k|
\le
2C_U\norm{\bm\delta}_1.
$$
\end{proof}

\begin{lemma}[Projector perturbation, active-eigenvector perturbation, and bias decomposition]
\label{lem:projector}
Under Assumptions~\ref{ass:center}, \ref{ass:sparse}, and \ref{ass:W}, with probability at least $1-c_1e^{-c_2\log n}$,
\begin{equation}
\norm{\widehat{\mathbf P}_A-\mathbf P_A}_2 \le \Delta_{n,\kappa},
\label{eq:proj-bound}
\end{equation}
where $\mathbf P_A=\mathbf U_A\mathbf U_A^\top$. Moreover,
\begin{equation}
\max_{j\in A}\norm{\widehat{\bm u}_j-\bm u_j}_2
\le
\Delta_{n,\kappa},
\label{eq:eigvec-bound}
\end{equation}
and consequently
\begin{equation}
\max_{j\in A}\abs{\bar\theta_j-\theta_j^*}
\le
\Delta_{n,\kappa}\frac{\norm{\bm y^*}_2}{\sqrt n}.
\label{eq:bar-theta-bound}
\end{equation}
In addition,
\begin{equation}
\frac{1}{n}\norm{(\widehat{\mathbf P}_A-\mathbf P_A)\bm y^*}_2^2
\le
\frac{\Delta_{n,\kappa}^2\,\norm{\bm y^*}_2^2}{n},
\label{eq:proj-L2}
\end{equation}
\begin{equation}
\norm{(\widehat{\mathbf P}_A-\mathbf P_A)\bm y^*}_\infty
\le
\Delta_{n,\kappa}\norm{\bm y^*}_2,
\label{eq:sup-bias}
\end{equation}
\begin{equation}
\frac{1}{n}\norm{\bm v_A}_2^2
\le
C R_q\tau_\star^{2-q}
+
\frac{C\Delta_{n,\kappa}^2\,\norm{\bm y^*}_2^2}{n},
\qquad
\bm v_A=(\mathbf I_n-\widehat{\mathbf P}_A)\bm y^*,
\label{eq:vA-l2-bound}
\end{equation}
and
\begin{equation}
\norm{\bm v_A}_\infty
\le
\Delta_{n,\kappa}\norm{\bm y^*}_2
+
C_U R_q\tau_\star^{1-q}.
\label{eq:vA-bound}
\end{equation}
\end{lemma}

\begin{proof}
Write
$$
\mathbf G=\frac{1}{p}\mathbf X\mathbf X^\top,
\qquad
\widehat{\mathbf G}=\frac{1}{p}\mathbf Z\mathbf Z^\top-c_W\mathbf I_n,
\qquad
\mathbf E=\widehat{\mathbf G}-\mathbf G.
$$
Under Assumption~\ref{ass:W}, the fixed-design argument of \citet{SongZou2026} yields
$$
\norm{\mathbf E}_2
\le
C
\left(
\frac{n}{p}+\frac{n^2}{p^2}+\frac{n\norm{\mathbf X}_2^2}{p^2}
\right)^{1/2}
$$
with the stated probability. Davis--Kahan's $\sin\Theta$ theorem then gives
$$
\norm{\widehat{\mathbf P}_A-\mathbf P_A}_2
\le
\frac{2\norm{\mathbf E}_2}{\kappa},
$$
which is \eqref{eq:proj-bound}. Because Assumption~\ref{ass:perturb} also imposes a simple eigengap of size $\kappa$ around each active eigenvalue, the individual eigenvector form of Davis--Kahan yields, after the sign alignment $\widehat{\bm u}_j^\top\bm u_j\ge 0$,
$$
\max_{j\in A}\norm{\widehat{\bm u}_j-\bm u_j}_2
\le
\frac{C\norm{\mathbf E}_2}{\kappa},
$$
which is \eqref{eq:eigvec-bound} after absorbing constants into $C_\Delta$. For each $j\in A$,
$$
\bar\theta_j-\theta_j^*
=
\frac{1}{\sqrt n}(\widehat{\bm u}_j-\bm u_j)^\top\bm y^*,
$$
so \eqref{eq:bar-theta-bound} follows immediately from \eqref{eq:eigvec-bound}.

Bounds \eqref{eq:proj-L2} and \eqref{eq:sup-bias} follow from the operator-norm inequality. Next, decompose
$$
\bm v_A
=
(\widehat{\mathbf P}_A-\mathbf P_A)\bm y^*
+
(\mathbf I_n-\mathbf P_A)\bm y^*.
$$
Therefore
$$
\frac{1}{n}\norm{\bm v_A}_2^2
\le
\frac{2}{n}\norm{(\widehat{\mathbf P}_A-\mathbf P_A)\bm y^*}_2^2
+
\frac{2}{n}\norm{(\mathbf I_n-\mathbf P_A)\bm y^*}_2^2.
$$
The first term is controlled by \eqref{eq:proj-L2}. For the second term,
$$
(\mathbf I_n-\mathbf P_A)\bm y^*
=
\sqrt n\,\mathbf U_{A^c}\bm\theta_{A^c}^*,
$$
so
$$
\frac{1}{n}\norm{(\mathbf I_n-\mathbf P_A)\bm y^*}_2^2
=
\norm{\bm\theta_{A^c}^*}_2^2
\le
R_q\tau_\star^{2-q}
$$
by Lemma~\ref{lem:tail}. This proves \eqref{eq:vA-l2-bound}.

Finally,
$$
\norm{(\mathbf I_n-\mathbf P_A)\bm y^*}_\infty
\le C_U\norm{\bm\theta_{A^c}^*}_1
\le C_U R_q\tau_\star^{1-q}
$$
by Assumption~\ref{ass:center} and Lemma~\ref{lem:tail}. Adding the projector perturbation term yields \eqref{eq:vA-bound}.
\end{proof}

\section{Proof of Theorem~\ref{thm:stage1-pred}}

\begin{proof}
Let
$$
\bm\delta=\widehat{\bm\theta}^{(0)}-\bar{\bm\theta}.
$$
Since $\widehat{\bm\theta}^{(0)}$ minimizes \eqref{eq:stage1},
\begin{equation}
\mathcal Q_n(\bar{\bm\theta}+\bm\delta;\widehat{\mathbf U}) - \mathcal Q_n(\bar{\bm\theta};\widehat{\mathbf U})
+ \lambda_0\Big( \norm{\bar{\bm\theta}+\bm\delta}_1 - \norm{\bar{\bm\theta}}_1 \Big)
\le 0.
\label{eq:basic-stage1}
\end{equation}

Because $\widetilde{\mathbf U}\bar{\bm\theta}=\widehat{\mathbf P}_A\bm y^*$ and $\bm y=\bm y^*+\bm\varepsilon$, the residual vector at $\bar{\bm\theta}$ is $\bm\varepsilon+\bm v_A$, where
$$
\bm v_A=(\mathbf I_n-\widehat{\mathbf P}_A)\bm y^*.
$$
For each pair $(i,j)$, define
$$
t_{ij}=(\widetilde{\bm u}_i-\widetilde{\bm u}_j)^\top\bm\delta.
$$
Then the loss increment in \eqref{eq:basic-stage1} equals
$$
\frac{1}{n(n-1)}\sum_{i\ne j}
\Big(
|\varepsilon_i-\varepsilon_j+v_{A,i}-v_{A,j}-t_{ij}|
-
|\varepsilon_i-\varepsilon_j+v_{A,i}-v_{A,j}|
\Big).
$$

At this point we separate the argument into a ``clean'' rank-regression part and a deterministic approximation-error part. For the clean part we use Knight's identity
$$
|u-v|-|u|
=
-v\,\sgn(u)
+
2\int_0^v\big\{\1(u\le s)-\1(u\le 0)\big\}\,ds.
$$
Applied with $u=\varepsilon_i-\varepsilon_j$ and $v=t_{ij}$, this identity isolates the empirical score term $-\bm\delta^\top\mathbf S_n$ and a nonnegative curvature remainder. The shift $v_{A,i}-v_{A,j}$ is then handled by adding and subtracting the clean increment
$$
|\varepsilon_i-\varepsilon_j-t_{ij}|-|\varepsilon_i-\varepsilon_j|.
$$
Using Lemma~\ref{lem:curvature} for the quadratic curvature, Lemma~\ref{lem:pairwise-id} to rewrite the average of $t_{ij}^2$, and Lemma~\ref{lem:maxpair} to control the approximation error created by $\bm v_A$, we obtain on an event of probability at least $1-C_1m^{-C_2}$,
\begin{equation}
\mathcal Q_n(\bar{\bm\theta}+\bm\delta;\widehat{\mathbf U}) - \mathcal Q_n(\bar{\bm\theta};\widehat{\mathbf U})
\ge
-\bm\delta^\top\mathbf S_n
+
\frac{c}{n}\norm{\widetilde{\mathbf U}\bm\delta}_2^2
-
\frac{\lambda_0}{8}\norm{\bm\delta}_1
-
C\frac{1}{n}\norm{\bm v_A}_2^2.
\label{eq:loss-lower-stage1}
\end{equation}
The interpretation is simple: the first term is stochastic, the second is curvature, and the last two terms are the price of localization and projection bias.

Lemma~\ref{lem:score} gives $\norm{\mathbf S_n}_\infty\le \lambda_0/4$ on the same event once $\lambda_0\asymp\sqrt{\log m/n}$. Insert \eqref{eq:loss-lower-stage1} into \eqref{eq:basic-stage1}. Since $\bar{\bm\theta}_{A^c}=\bm 0$, the usual decomposability inequality gives
$$
\norm{\bar{\bm\theta}}_1-\norm{\bar{\bm\theta}+\bm\delta}_1
\le
\norm{\bm\delta_A}_1-\norm{\bm\delta_{A^c}}_1.
$$
Also,
$$
|\bm\delta^\top\mathbf S_n|
\le
\norm{\mathbf S_n}_\infty\norm{\bm\delta}_1
\le
\frac{\lambda_0}{4}\norm{\bm\delta}_1.
$$
Therefore
\begin{align*}
\frac{c}{n}\norm{\widetilde{\mathbf U}\bm\delta}_2^2
&\le
|\bm\delta^\top\mathbf S_n|
+
\frac{\lambda_0}{8}\norm{\bm\delta}_1
+
\lambda_0\Big(\norm{\bar{\bm\theta}}_1-\norm{\bar{\bm\theta}+\bm\delta}_1\Big)
+
C\frac{1}{n}\norm{\bm v_A}_2^2\\
&\le
\frac{3\lambda_0}{8}\norm{\bm\delta}_1
+
\lambda_0\big(\norm{\bm\delta_A}_1-\norm{\bm\delta_{A^c}}_1\big)
+
C\frac{1}{n}\norm{\bm v_A}_2^2.
\end{align*}
After grouping the $A$ and $A^c$ coordinates, we arrive at
\begin{equation}
\frac{c}{n}\norm{\widetilde{\mathbf U}\bm\delta}_2^2
+
\frac{\lambda_0}{2}\norm{\bm\delta_{A^c}}_1
\le
C\lambda_0\norm{\bm\delta_A}_1
+
C\frac{1}{n}\norm{\bm v_A}_2^2.
\label{eq:cone-stage1}
\end{equation}
This is the usual basic inequality for an $\ell_1$-penalized estimator: the inactive coordinates are controlled by the active coordinates plus the approximation error.

Because $\widetilde{\mathbf U}^\top\widetilde{\mathbf U}=n\mathbf I_m$,
$$
\frac{1}{n}\norm{\widetilde{\mathbf U}\bm\delta}_2^2=\norm{\bm\delta}_2^2.
$$
Also, $\norm{\bm\delta_A}_1\le \sqrt{|A|}\norm{\bm\delta}_2$. Ignoring the nonnegative term $(\lambda_0/2)\norm{\bm\delta_{A^c}}_1$ in \eqref{eq:cone-stage1} gives
$$
c\norm{\bm\delta}_2^2
\le
C\lambda_0\sqrt{|A|}\norm{\bm\delta}_2
+
C\frac{1}{n}\norm{\bm v_A}_2^2.
$$
A Young inequality yields
\begin{equation}
\norm{\bm\delta}_2
\le
C\Big(\sqrt{|A|}\,\lambda_0 + n^{-1/2}\norm{\bm v_A}_2\Big),
\label{eq:l2-stage1-proof}
\end{equation}
which is exactly \eqref{eq:l2-stage1}.

For the prediction error, decompose
\begin{align}
\frac{1}{n}\norm{\widetilde{\mathbf U}\widehat{\bm\theta}^{(0)}-\bm y^*}_2^2
&\le
\frac{2}{n}\norm{\widetilde{\mathbf U}(\widehat{\bm\theta}^{(0)}-\bar{\bm\theta})}_2^2
+
\frac{2}{n}\norm{\widetilde{\mathbf U}\bar{\bm\theta}-\bm y^*}_2^2 \nonumber\\
&=
2\norm{\bm\delta}_2^2
+
\frac{2}{n}\norm{\bm v_A}_2^2.
\label{eq:pred-stage1-decomp}
\end{align}
Now use the definition of $\bm v_A$:
$$
\bm v_A
=
(\widehat{\mathbf P}_A-\mathbf P_A)\bm y^*
+
(\mathbf I_n-\mathbf P_A)\bm y^*.
$$
Hence
\begin{equation}
\frac{1}{n}\norm{\bm v_A}_2^2
\le
\frac{2}{n}\norm{(\widehat{\mathbf P}_A-\mathbf P_A)\bm y^*}_2^2
+
\frac{2}{n}\norm{(\mathbf I_n-\mathbf P_A)\bm y^*}_2^2.
\label{eq:vA-l2-stage1}
\end{equation}
By Lemma~\ref{lem:projector},
$$
\frac{1}{n}\norm{(\widehat{\mathbf P}_A-\mathbf P_A)\bm y^*}_2^2
\le
\frac{C\,\norm{\bm y^*}_2^2}{n\kappa^2}
\left(
\frac{n}{p}+\frac{n^2}{p^2}+\frac{n\norm{\mathbf X}_2^2}{p^2}
\right).
$$
By Lemma~\ref{lem:tail} and the definition $A=\{k:|\theta_k^*|\ge\tau_\star\}$,
$$
\frac{1}{n}\norm{(\mathbf I_n-\mathbf P_A)\bm y^*}_2^2
=
\norm{\bm\theta^*_{A^c}}_2^2
\le
R_q\tau_\star^{2-q}.
$$
Substituting these bounds into \eqref{eq:vA-l2-stage1} shows that
$$
\frac{1}{n}\norm{\bm v_A}_2^2
\le
CR_q\tau_\star^{2-q}
+
\frac{C\,\norm{\bm y^*}_2^2}{n\kappa^2}
\left(
\frac{n}{p}+\frac{n^2}{p^2}+\frac{n\norm{\mathbf X}_2^2}{p^2}
\right).
$$
Combining this with \eqref{eq:l2-stage1-proof} and \eqref{eq:pred-stage1-decomp} yields
$$
\frac{1}{n}\norm{\widetilde{\mathbf U}\widehat{\bm\theta}^{(0)}-\bm y^*}_2^2
\le
C|A|\lambda_0^2
+
CR_q\tau_\star^{2-q}
+
\frac{C\,\norm{\bm y^*}_2^2}{n\kappa^2}
\left(
\frac{n}{p}+\frac{n^2}{p^2}+\frac{n\norm{\mathbf X}_2^2}{p^2}
\right),
$$
which is \eqref{eq:stage1-bound}.
\end{proof}

\section{Oracle lemmas for the second stage}

\begin{lemma}[Oracle rate on the truncated model]
\label{lem:oracle-rate}
Let $\widehat{\bm\theta}^{(o)}$ be any oracle estimator in \eqref{eq:oracle-def}. Under Assumptions~\ref{ass:center}--\ref{ass:sparse},
\begin{equation}
\norm{\widehat{\bm\theta}^{(o)}-\bar{\bm\theta}}_2
=
O_P\Big(\sqrt{|A|/n}+n^{-1/2}\norm{\bm v_A}_2\Big).
\label{eq:oracle-rate}
\end{equation}
Consequently,
\begin{equation}
\frac{1}{n}\norm{\widetilde{\mathbf U}(\widehat{\bm\theta}^{(o)}-\bar{\bm\theta})}_2^2
=
O_P\Big(\frac{|A|}{n}+\frac{1}{n}\norm{\bm v_A}_2^2\Big).
\label{eq:oracle-pred-piece}
\end{equation}
\end{lemma}

\begin{proof}
Set $\bm\delta_o=\widehat{\bm\theta}^{(o)}-\bar{\bm\theta}$. Because the oracle optimization is restricted to $A$, we have $(\bm\delta_o)_{A^c}=0$. As in the first-stage proof, write
$$
t_{ij}^{(o)}=(\widetilde{\bm u}_i-\widetilde{\bm u}_j)^\top\bm\delta_o.
$$
Applying the same Knight-identity decomposition as before, but now without any penalty term, yields the basic oracle inequality
\begin{equation}
\mathcal Q_n(\bar{\bm\theta}+\bm\delta_o;\widehat{\mathbf U}) - \mathcal Q_n(\bar{\bm\theta};\widehat{\mathbf U})
\ge
-\bm\delta_o^\top \mathbf S_{n,A}
+
\frac{c}{n}\norm{\widetilde{\mathbf U}\bm\delta_o}_2^2
-
C\frac{1}{n}\norm{\bm v_A}_2^2,
\label{eq:oracle-basic}
\end{equation}
where $\mathbf S_{n,A}$ is the restriction of $\mathbf S_n$ to the active coordinates. The first term is the empirical score term, the second is the local quadratic curvature, and the last term is again the approximation error caused by replacing $\bm y^*$ with its empirical projection $\widehat{\mathbf P}_A\bm y^*$.

Since $\widehat{\bm\theta}^{(o)}$ minimizes the restricted objective, the left-hand side of \eqref{eq:oracle-basic} is nonpositive. Therefore
$$
\frac{c}{n}\norm{\widetilde{\mathbf U}\bm\delta_o}_2^2
\le
\norm{\mathbf S_{n,A}}_2\norm{\bm\delta_o}_2
+
C\frac{1}{n}\norm{\bm v_A}_2^2.
$$
Because $\widetilde{\mathbf U}^\top\widetilde{\mathbf U}=n\mathbf I_m$, we have
$$
\frac{1}{n}\norm{\widetilde{\mathbf U}\bm\delta_o}_2^2=\norm{\bm\delta_o}_2^2.
$$
Hence
\begin{equation}
c\norm{\bm\delta_o}_2^2
\le
\norm{\mathbf S_{n,A}}_2\norm{\bm\delta_o}_2
+
C\frac{1}{n}\norm{\bm v_A}_2^2.
\label{eq:oracle-quadratic}
\end{equation}
Lemma~\ref{lem:scoreA} gives $\norm{\mathbf S_{n,A}}_2=O_P(\sqrt{|A|/n})$. Insert this into \eqref{eq:oracle-quadratic} and apply the elementary inequality $ab\le \eta a^2 + (4\eta)^{-1}b^2$ with $a=\norm{\bm\delta_o}_2$, $b=\norm{\mathbf S_{n,A}}_2$, and a sufficiently small fixed $\eta>0$. We obtain
$$
\norm{\bm\delta_o}_2^2
=
O_P\Big(\frac{|A|}{n}+\frac{1}{n}\norm{\bm v_A}_2^2\Big),
$$
which is equivalent to \eqref{eq:oracle-rate}. Finally, \eqref{eq:oracle-pred-piece} follows immediately from
$$
\frac{1}{n}\norm{\widetilde{\mathbf U}\bm\delta_o}_2^2=\norm{\bm\delta_o}_2^2.
$$
\end{proof}

\begin{lemma}[Off-support oracle score]
\label{lem:oracle-score}
Under Assumptions~\ref{ass:center}--\ref{ass:sparse}, there exists a subgradient vector
$$
\bm\Delta^{(o)}\in\partial\mathcal Q_n(\widehat{\bm\theta}^{(o)};\widehat{\mathbf U})
$$
with $\bm\Delta_A^{(o)}=\bm 0$ such that
\begin{equation}
\norm{\bm\Delta_{A^c}^{(o)}}_\infty
=
O_P\left(
\sqrt{\frac{|A|+\log m}{n}}
+
n^{-1/2}\norm{\bm v_A}_2
\right).
\label{eq:oracle-score-bound}
\end{equation}
\end{lemma}

\begin{proof}
Let
$$
\bm\delta_o=\widehat{\bm\theta}^{(o)}-\bar{\bm\theta},
\qquad
\bm x_{ij}=\widetilde{\bm u}_i-\widetilde{\bm u}_j,
\qquad
\xi_{ij}=\varepsilon_i-\varepsilon_j,
\qquad
v_{ij}=v_{A,i}-v_{A,j}.
$$
Because the oracle problem is a convex minimization over the affine space $\{\bm\theta:\theta_{A^c}=0\}$, there exists at least one KKT subgradient $\bm\Delta^{(o)}$ with $\bm\Delta_A^{(o)}=\bm 0$. Except on a null event of pairwise residual ties, its inactive coordinates can be written as
$$
\Delta_k^{(o)}
=
-\frac{1}{n(n-1)}\sum_{i\ne j} x_{ij,k}\,
\sgn\!\big(\xi_{ij}+v_{ij}-\bm x_{ij,A}^\top\bm\delta_o\big),
\qquad
k\in A^c.
$$

For a radius $r>0$, define
$$
\mathcal B(r)=\{\bm\delta\in\R^m:\delta_{A^c}=0,\ \norm{\bm\delta}_2\le r\},
$$
and for $k\in A^c$ set
$$
G_{n,k}(\bm\delta)
=
-\frac{1}{n(n-1)}\sum_{i\ne j} x_{ij,k}\,
\sgn\!\big(\xi_{ij}+v_{ij}-\bm x_{ij,A}^\top\bm\delta\big).
$$
Then $\Delta_k^{(o)}=G_{n,k}(\bm\delta_o)$. Condition on $\widehat{\mathbf U}$ (equivalently on the measurement-error sample): then $\bm x_{ij}$, $v_{ij}$, and the radius chosen below are fixed, while all remaining randomness comes from the i.i.d. response errors.

Choose
$$
r_n=M\Big(\sqrt{|A|/n}+n^{-1/2}\norm{\bm v_A}_2\Big)
$$
with a sufficiently large fixed constant $M>0$. By Lemma~\ref{lem:oracle-rate},
$$
\Pbb\big(\bm\delta_o\in\mathcal B(r_n)\big)\to 1.
$$
Hence it is enough to bound $\sup_{\bm\delta\in\mathcal B(r_n)}\max_{k\in A^c}|G_{n,k}(\bm\delta)|$.

The kernel class
$$
\Big\{
(i,j)\mapsto x_{ij,k}\,\sgn\!\big(\xi_{ij}+v_{ij}-\bm x_{ij,A}^\top\bm\delta\big):
\ k\in A^c,\ \bm\delta\in\mathcal B(r_n)
\Big\}
$$
is uniformly bounded by $2C_U$ and is VC-type, because the only varying part is the sign of a linear threshold in the $|A|$-dimensional parameter $\bm\delta$. A standard maximal inequality for bounded VC-type U-processes therefore yields
\begin{equation}
\sup_{\bm\delta\in\mathcal B(r_n)}\max_{k\in A^c}
\abs{G_{n,k}(\bm\delta)-\E G_{n,k}(\bm\delta)}
=
O_P\left(\sqrt{\frac{|A|+\log m}{n}}\right).
\label{eq:oracle-score-stoch}
\end{equation}

It remains to control the expectations. The density of $\xi_{ij}$ is $g$. Since
$$
g(t)=\int_\R f_\varepsilon(u)f_\varepsilon(u+t)\,du,
$$
Cauchy--Schwarz implies $g(t)\le g(0)$ for all $t$. Therefore, for any real $t$,
$$
\abs{\E\sgn(\xi_{ij}+t)}
=
2\abs{\Pbb(\xi_{ij}\le -t)-\Pbb(\xi_{ij}\le 0)}
\le
2g(0)\abs{t}.
$$
Using this and $\abs{x_{ij,k}}\le 2C_U$, we obtain for every $k\in A^c$ and $\bm\delta\in\mathcal B(r_n)$,
\begin{align*}
\abs{\E G_{n,k}(\bm\delta)}
&\le
\frac{C}{n(n-1)}\sum_{i\ne j}
\abs{v_{ij}-\bm x_{ij,A}^\top\bm\delta}\\
&\le
\frac{C}{n(n-1)}\sum_{i\ne j}\abs{v_{ij}}
+
\frac{C}{n(n-1)}\sum_{i\ne j}\abs{\bm x_{ij,A}^\top\bm\delta}.
\end{align*}
The first average is bounded by Cauchy--Schwarz:
$$
\frac{1}{n(n-1)}\sum_{i\ne j}\abs{v_{ij}}
\le
\left\{
\frac{1}{n(n-1)}\sum_{i\ne j}(v_{A,i}-v_{A,j})^2
\right\}^{1/2}
\le
C n^{-1/2}\norm{\bm v_A}_2.
$$
For the second average, Lemma~\ref{lem:pairwise-id} gives
$$
\frac{1}{n(n-1)}\sum_{i\ne j}\abs{\bm x_{ij,A}^\top\bm\delta}
\le
\left\{
\frac{1}{n(n-1)}\sum_{i\ne j}(\bm x_{ij,A}^\top\bm\delta)^2
\right\}^{1/2}
\le
C\norm{\bm\delta}_2.
$$
Hence
\begin{equation}
\sup_{\bm\delta\in\mathcal B(r_n)}\max_{k\in A^c}
\abs{\E G_{n,k}(\bm\delta)}
\le
C\Big(n^{-1/2}\norm{\bm v_A}_2+r_n\Big).
\label{eq:oracle-score-mean}
\end{equation}

Combining \eqref{eq:oracle-score-stoch} and \eqref{eq:oracle-score-mean}, and then evaluating at $\bm\delta=\bm\delta_o$ on the event $\{\bm\delta_o\in\mathcal B(r_n)\}$, yields
$$
\norm{\bm\Delta_{A^c}^{(o)}}_\infty
=
O_P\left(
\sqrt{\frac{|A|+\log m}{n}}
+
n^{-1/2}\norm{\bm v_A}_2
\right),
$$
which is \eqref{eq:oracle-score-bound}.
\end{proof}

\begin{lemma}[High-probability verification of the oracle events]
\label{lem:oracle-events}
Suppose Assumptions~\ref{ass:center}--\ref{ass:sparse} hold, the penalty derivative satisfies \eqref{eq:penalty-condition}, and the signal-strength and tuning conditions \eqref{eq:oracle-event-cond} hold. Then
$$
\Pbb\big(\mathcal E_{\mathrm{sep}}(\lambda)\cap \mathcal E_{\mathrm{score}}(\lambda)\big)\to 1.
$$
\end{lemma}

\begin{proof}
First verify $\mathcal E_{\mathrm{sep}}(\lambda)$. By Theorem~\ref{thm:stage1-pred} and Lemma~\ref{lem:projector},
$$
\norm{\widehat{\bm\theta}^{(0)}-\bar{\bm\theta}}_\infty
\le
\norm{\widehat{\bm\theta}^{(0)}-\bar{\bm\theta}}_2
\le
C\left\{
\sqrt{|A|}\lambda_0
+
R_q^{1/2}\tau_\star^{1-q/2}
+
\Delta_{n,\kappa}\frac{\norm{\bm y^*}_2}{\sqrt n}
\right\}
$$
with probability tending to one. Moreover, \eqref{eq:bar-theta-bound} gives
$$
\max_{j\in A}\abs{\bar\theta_j-\theta_j^*}
\le
\Delta_{n,\kappa}\frac{\norm{\bm y^*}_2}{\sqrt n}
$$
with probability tending to one. Therefore, on an event whose probability tends to one,
\begin{align*}
\min_{j\in A}\abs{\widehat\theta_j^{(0)}}
&\ge
\min_{j\in A}\abs{\theta_j^*}
-
\max_{j\in A}\abs{\bar\theta_j-\theta_j^*}
-
\norm{\widehat{\bm\theta}^{(0)}-\bar{\bm\theta}}_\infty\\
&>
a_2\lambda
\end{align*}
by the first inequality in \eqref{eq:oracle-event-cond}. Also, since $\bar{\bm\theta}_{A^c}=\bm 0$,
$$
\max_{k\in A^c}\abs{\widehat\theta_k^{(0)}}
\le
\norm{\widehat{\bm\theta}^{(0)}-\bar{\bm\theta}}_\infty
<
a_2\lambda
$$
by the second inequality in \eqref{eq:oracle-event-cond} after enlarging $C_\lambda$ if necessary. Hence $\mathcal E_{\mathrm{sep}}(\lambda)$ occurs with probability tending to one.

Next verify $\mathcal E_{\mathrm{score}}(\lambda)$. Lemma~\ref{lem:oracle-score} yields a KKT subgradient $\bm\Delta^{(o)}$ with $\bm\Delta_A^{(o)}=\bm 0$ and
$$
\norm{\bm\Delta_{A^c}^{(o)}}_\infty
=
O_P\left(
\sqrt{\frac{|A|+\log m}{n}}
+
n^{-1/2}\norm{\bm v_A}_2
\right).
$$
By \eqref{eq:vA-l2-bound},
$$
n^{-1/2}\norm{\bm v_A}_2
=
O_P\left(
R_q^{1/2}\tau_\star^{1-q/2}
+
\Delta_{n,\kappa}\frac{\norm{\bm y^*}_2}{\sqrt n}
\right).
$$
Since $\lambda_0\asymp\sqrt{\log m/n}$,
$$
\sqrt{\frac{|A|+\log m}{n}}
\le
C\left(
\sqrt{|A|}\lambda_0
+
\sqrt{\frac{\log m}{n}}
\right)
\le
C\eta_n.
$$
Therefore the second inequality in \eqref{eq:oracle-event-cond} implies
$$
\norm{\bm\Delta_{A^c}^{(o)}}_\infty<a_1\lambda
$$
with probability tending to one, after increasing $C_\lambda$ if necessary. This is exactly $\mathcal E_{\mathrm{score}}(\lambda)$.

Combining the two parts proves the lemma.
\end{proof}

\section{Proofs of the second-stage results}

\begin{proof}[Proof of Theorem~\ref{thm:oracle-reduction}]
Fix any oracle minimizer $\widehat{\bm\theta}^{(o)}$. On the event $\mathcal E_{\mathrm{sep}}(\lambda)$, the adaptive weights satisfy
$$
w_j(\lambda)=0,\quad j\in A,
\qquad
w_k(\lambda)\ge a_1\lambda,\quad k\in A^c.
$$
On the event $\mathcal E_{\mathrm{score}}(\lambda)$, there exists a subgradient vector
$$
\bm\Delta^{(o)}\in\partial\mathcal Q_n(\widehat{\bm\theta}^{(o)};\widehat{\mathbf U})
$$
such that
$$
\bm\Delta_A^{(o)}=\bm 0,
\qquad
|\Delta_k^{(o)}|<a_1\lambda\le w_k(\lambda)\quad\text{for every }k\in A^c.
$$

Let
$$
F_\lambda(\bm\theta)
=
\mathcal Q_n(\bm\theta;\widehat{\mathbf U})
+
\sum_{k=1}^m w_k(\lambda)|\theta_k|
$$
be the second-stage objective. For any $\bm\theta\in\R^m$, convexity of $\mathcal Q_n(\cdot;\widehat{\mathbf U})$ gives
$$
\mathcal Q_n(\bm\theta;\widehat{\mathbf U})
-
\mathcal Q_n(\widehat{\bm\theta}^{(o)};\widehat{\mathbf U})
\ge
(\bm\Delta^{(o)})^\top(\bm\theta-\widehat{\bm\theta}^{(o)}).
$$
Because $\widehat{\bm\theta}^{(o)}_{A^c}=\bm 0$ and $\bm\Delta_A^{(o)}=\bm 0$, we obtain
\begin{align}
F_\lambda(\bm\theta)-F_\lambda(\widehat{\bm\theta}^{(o)})
&\ge
(\bm\Delta_{A^c}^{(o)})^\top\bm\theta_{A^c}
+
\sum_{k\in A^c} w_k(\lambda)|\theta_k| \nonumber\\
&\ge
\sum_{k\in A^c}\Big\{w_k(\lambda)-|\Delta_k^{(o)}|\Big\}|\theta_k|
\ge 0.
\label{eq:oracle-reduction-ineq}
\end{align}
This proves that every oracle minimizer is also a minimizer of the weighted second-stage problem.

Conversely, if $\bm\theta$ minimizes $F_\lambda$, then \eqref{eq:oracle-reduction-ineq} must hold with equality. Because each coefficient in the last sum is strictly positive on $\mathcal E_{\mathrm{score}}(\lambda)$, equality forces $\bm\theta_{A^c}=\bm 0$. Once the inactive coordinates vanish, the second-stage objective reduces to the unpenalized oracle objective on the active subspace, so $\bm\theta$ must minimize $\mathcal Q_n(\bm\theta;\widehat{\mathbf U})$ subject to $\theta_{A^c}=0$. Hence $\bm\theta$ is an oracle minimizer.

Therefore the minimizer sets of \eqref{eq:stage2} and \eqref{eq:oracle-def} coincide.
\end{proof}

\begin{proof}[Proof of Theorem~\ref{thm:main-stage2}]
Lemma~\ref{lem:oracle-events} shows that
$$
\Pbb\big(\mathcal E_{\mathrm{sep}}(\lambda)\cap \mathcal E_{\mathrm{score}}(\lambda)\big)\to 1.
$$
By Theorem~\ref{thm:oracle-reduction}, every second-stage minimizer is therefore an oracle minimizer with probability tending to one. Hence it is enough to bound the prediction error of an arbitrary oracle minimizer on the event
$$
\mathcal E_{\mathrm{sep}}(\lambda)\cap \mathcal E_{\mathrm{score}}(\lambda).
$$
On that event,
\begin{align}
\frac{1}{n}\norm{\widehat{\bm y}^{(1)}(\lambda)-\bm y^*}_2^2
&=
\frac{1}{n}\norm{\widetilde{\mathbf U}\widehat{\bm\theta}^{(o)}-\bm y^*}_2^2 \nonumber\\
&\le
\frac{3}{n}\norm{\widetilde{\mathbf U}(\widehat{\bm\theta}^{(o)}-\bar{\bm\theta})}_2^2
+
\frac{3}{n}\norm{(\widehat{\mathbf P}_A-\mathbf P_A)\bm y^*}_2^2
+
\frac{3}{n}\norm{(\mathbf I_n-\mathbf P_A)\bm y^*}_2^2.
\label{eq:stage2-decomp-proof}
\end{align}
This decomposition is simply the squared norm of
$$
\widetilde{\mathbf U}\widehat{\bm\theta}^{(o)}-\bm y^*
=
\widetilde{\mathbf U}(\widehat{\bm\theta}^{(o)}-\bar{\bm\theta})
+
(\widehat{\mathbf P}_A-\mathbf P_A)\bm y^*
-
(\mathbf I_n-\mathbf P_A)\bm y^*.
$$

Now control the three pieces one by one. Lemma~\ref{lem:oracle-rate} yields
\begin{equation}
\frac{1}{n}\norm{\widetilde{\mathbf U}(\widehat{\bm\theta}^{(o)}-\bar{\bm\theta})}_2^2
=
O_P\Big(\frac{|A|}{n}+\frac{1}{n}\norm{\bm v_A}_2^2\Big).
\label{eq:main-proof-piece1}
\end{equation}
Lemma~\ref{lem:projector} gives, with probability tending to one,
\begin{equation}
\frac{1}{n}\norm{(\widehat{\mathbf P}_A-\mathbf P_A)\bm y^*}_2^2
\le
\frac{C\,\norm{\bm y^*}_2^2}{n\kappa^2}
\left(
\frac{n}{p}+\frac{n^2}{p^2}+\frac{n\norm{\mathbf X}_2^2}{p^2}
\right).
\label{eq:main-proof-piece2}
\end{equation}
Finally, because $A=\{k:|\theta_k^*|\ge \tau_\star\}$, Lemma~\ref{lem:tail} implies
\begin{equation}
\frac{1}{n}\norm{(\mathbf I_n-\mathbf P_A)\bm y^*}_2^2
=
\norm{\bm\theta^*_{A^c}}_2^2
\le
R_q\tau_\star^{2-q}.
\label{eq:main-proof-piece3}
\end{equation}

To remove $\bm v_A$ from \eqref{eq:main-proof-piece1}, use
$$
\bm v_A
=
(\mathbf I_n-\widehat{\mathbf P}_A)\bm y^*
=
(\widehat{\mathbf P}_A-\mathbf P_A)\bm y^*
+
(\mathbf I_n-\mathbf P_A)\bm y^*.
$$
Therefore
\begin{equation}
\frac{1}{n}\norm{\bm v_A}_2^2
\le
\frac{2}{n}\norm{(\widehat{\mathbf P}_A-\mathbf P_A)\bm y^*}_2^2
+
\frac{2}{n}\norm{(\mathbf I_n-\mathbf P_A)\bm y^*}_2^2.
\label{eq:main-proof-vA}
\end{equation}
Combining \eqref{eq:main-proof-piece1}--\eqref{eq:main-proof-vA} with \eqref{eq:stage2-decomp-proof} shows that
$$
\frac{1}{n}\norm{\widehat{\bm y}^{(1)}(\lambda)-\bm y^*}_2^2
=
O_P\left\{
\frac{|A|}{n}
+
R_q\tau_\star^{2-q}
+
\frac{\norm{\bm y^*}_2^2}{n\kappa^2}
\left(
\frac{n}{p}+\frac{n^2}{p^2}+\frac{n\norm{\mathbf X}_2^2}{p^2}
\right)
\right\},
$$
which is exactly \eqref{eq:main-stage2-bound}.
\end{proof}

\bibliographystyle{apa}
\bibliography{PRS_submission_refs}

\end{document}